\documentclass[onecolumn,11pt,draftcls]{IEEEtran} 
\usepackage{amsbsy}
\usepackage{floatflt} 

\usepackage{amsmath}
\usepackage{amssymb}
\usepackage{times}
\usepackage{graphics}
\usepackage{graphicx}
\usepackage{xspace}
\usepackage{paralist} 
\usepackage{setspace} 
\usepackage{xypic}
\xyoption{curve}
\usepackage{latexsym}
\usepackage{theorem}
\usepackage{ifthen}
\usepackage{subfigure}
\usepackage{booktabs}
\usepackage{algorithm}
\usepackage{algorithmic}
\usepackage{color}

\newcommand{\mypar}[1]{\vspace{0.03in}\noindent{\bf #1.}}

{\theoremheaderfont{\it} \theorembodyfont{\rmfamily}
\newtheorem{theorem}{Theorem}
\newtheorem{lemma}[theorem]{Lemma}
\newtheorem{definition}[theorem]{Definition}

\newtheorem{assumption}[theorem]{Assumption}

}

\title{Cooperative Convex Optimization in Networked Systems: Augmented Lagrangian Algorithms with Directed Gossip Communication}
\author{Du$\breve{\mbox{s}}$an Jakoveti\'c, Jo\~ao Xavier, and Jos\'e M.~F.~Moura$^{\star}$
\thanks{The first and second authors are with the Instituto de Sistemas e Rob\'otica~(ISR), Instituto Superior T\'ecnico~(IST), 1049-001 Lisboa, Portugal. The first and third authors are with the Department of Electrical and Computer Engineering,
Carnegie Mellon University, Pittsburgh, PA 15213, USA (e-mail:
[djakovetic,jxavier]@isr.ist.utl.pt, moura@ece.cmu.edu, ph: (412)268-6341, fax: (412)268-3890.) This work is partially supported by: the Carnegie Mellon$|$Portugal Program under a grant from the Funda\c{c}\~ao de Ci$\hat{\mbox{e}}$ncia e Tecnologia~(FCT) from Portugal; by FCT grants SIPM PTDC/EEA-ACR/73749/2006 and SFRH/BD/33520/2008 (through the Carnegie Mellon/Portugal Program
managed by ICTI); by
ISR/IST plurianual funding (POSC program, FEDER); by AFOSR grant~\#~FA95501010291; and by NSF grant~\#~CCF1011903. Du$\breve{\mbox{s}}$an Jakoveti\'c holds a fellowship from~FCT.}}
\begin{document}
\maketitle \thispagestyle{empty} \maketitle
\vspace{-9mm}
\begin{abstract}
We study distributed optimization in networked systems, where nodes cooperate to
find the optimal quantity of common interest, $x=x^\star$. The objective function of the corresponding
 optimization problem is the sum of private (known only by a node,) convex, nodes' objectives and each node imposes a
 private convex constraint on the allowed values of $x$. We solve this problem for generic connected
  network topologies with asymmetric random link failures with a novel distributed, decentralized algorithm. We
   refer to this algorithm as \textbf{AL--G} (augmented Lagrangian gossiping,)
     and to its variants as \textbf{AL--MG} (augmented Lagrangian multi neighbor gossiping)
      and \textbf{AL--BG} (augmented Lagrangian broadcast gossiping.) The \textbf{AL--G} algorithm is based on the augmented Lagrangian dual function. Dual variables are updated by the standard method of
  multipliers, at a slow time scale. To update the primal variables, we propose a novel, Gauss-Seidel type, randomized algorithm, at a fast time scale.
  \textbf{AL--G} uses unidirectional gossip communication, only between immediate neighbors in the network and is resilient to random link failures. For networks with reliable communication (i.e., no failures,) the simplified, \textbf{AL--BG} (augmented Lagrangian broadcast gossiping) algorithm reduces communication, computation and data storage cost. We prove convergence for all proposed algorithms
    and demonstrate by simulations the effectiveness on two applications:
     $l_1$--regularized logistic regression for classification and cooperative spectrum
  sensing for cognitive radio networks.

\end{abstract}
\hspace{.43cm}\textbf{Keywords:}
Distributed algorithm, convex optimization, augmented Lagrangian, gossip communication
\newpage

\maketitle \thispagestyle{empty} \maketitle
%
%
%
\vspace{-3mm}
\section{Introduction}
\label{section_intro}
Recently, there has been
increased interest in large scale networked systems including networks of agents, wireless ad-hoc networks, and wireless sensor networks (WSNs.) Typically, these systems lack a central unit, and the inter-node communication is prone to random failures (e.g., random packet dropouts in WSNs.) In this paper, we consider a generic computational model that captures many applications in networked systems.
With this model, nodes cooperate to find the optimal parameter (scalar or vector) of common interest, $x=x^\star$, e.g., the optimal operating point of the network.
 Each node $i$ has a \emph{private} (known only at node $i$) cost function of $x$, e.g., a loss at node $i$ if operating at $x$. The total cost
  is the sum over the individual nodes' costs. Also, each node imposes a \emph{private} constraint on the allowed values of $x$ (e.g., allowed operating points
 at node $i$.) Applications of this computational model include resource allocation in wireless systems~\cite{boyd-johansson}, distributed estimation in wireless sensor networks,~\cite{SoummyaEst}, and distributed, cooperative spectrum sensing in cognitive radio networks,~\cite{bazerque_sensing, bazerque_lasso}.

More formally, nodes cooperatively solve the following optimization problem:
\begin{equation}
\begin{array}[+]{ll}
\mbox{minimize} & \sum_{i=1}^N f_i(x) \\
\mbox{subject to} & x \in \mathcal{X}_i,\,\,i=1,...,N
\end{array}.
\label{eqn_original_primal_problem}
\end{equation}
Here $N$ is the number of nodes in the network, the private cost functions $f_i: {\mathbb R^m} \rightarrow {\mathbb R}$ are convex, and each $f_i(\cdot)$ is known locally only by node $i$. The sets $\mathcal{X}_i$ are \emph{private}, closed, convex  constraint sets. We remark that~\eqref{eqn_original_primal_problem} captures the scenario when, in addition to private constraints, there is a public constraint $x \in X$ (where $X$ is a closed, convex set,) just by replacing $\mathcal{X}_i$ with $\mathcal{X}_i \cap X$.

This paper proposes a novel augmented Lagrangian (AL) primal-dual distributed algorithm for solving~\eqref{eqn_original_primal_problem}, which handles private costs $f_i(\cdot)$, private constraints
$\mathcal{X}_i$, and is resilient to random communication failures. We refer to this algorithm as \textbf{AL--G} (augmented Lagrangian gossiping.) We also consider two variants to \textbf{AL--G}, namely, the \textbf{AL--MG} (augmented Lagrangian multiple neighbor gossiping) and the \textbf{AL--BG} (augmented Lagrangian broadcast gossiping.) The \textbf{AL--G} and \textbf{AL--MG} algorithms use \emph{unidirectional} gossip communication (see, e.g.,~\cite{BoydGossip}). For networks with reliable communication (i.e., no failures,) we propose the simplified \textbf{AL--BG} algorithm with reduced communication, reduced computation, and lower data storage cost. Our algorithms update the dual variables by the standard method of multipliers,~\cite{BertsekasOptimization},
synchronously, at a slow time scale, and update the primal variables with a \emph{novel,} Gauss-Seidel type (see, e.g.,~\cite{BertsekasBook}) randomized algorithm with asynchronous gossip communication, at a fast
time scale. Proof of convergence for the method of multipliers (for the dual variables update)
is available in the literature, e.g.,~\cite{BertsekasOptimization}. However, our algorithms to update
 primal variables (referred to as P--AL--G (primal AL gossip), P--AL--MG and
 P--AL--BG) are novel, a major contribution of this paper is to prove convergence of the
  P--AL--G, for private constraints, under very generic network topologies, random link failures, and gossip communication.
   The proof is then adapted to P--AL--MG and P--AL--BG.

   The \textbf{AL-G} (and its variants \textbf{AL-MG} and \textbf{AL-BG}) algorithms
   are generic tools that fit many applications in networked systems. We provide
    two simulation examples, namely, $l_1$--regularized logistic regression for classification and cooperative spectrum sensing
    for cognitive radio networks. These simulation examples: 1) corroborate convergence of the proposed
    algorithms; and 2) compare their performance, in terms of communication and computational cost,
     with the algorithms in~\cite{asu-random,nedic_novo,bazerque_lasso,bazerque_sensing}.

\mypar{Comparison with existing work} We now identify
 important dimensions of the communication and computation models
 that characterize existing references
 and that help to contrast our paper with the relevant literature.
\emph{Optimization algorithms} to solve~\eqref{eqn_original_primal_problem}, or problems similar to~\eqref{eqn_original_primal_problem}, in a
 distributed way, are usually either primal-dual distributed algorithms or
 primal subgradient~algorithms.
\emph{Optimization constraints} on problem~\eqref{eqn_original_primal_problem} can either be:
no constraints ($\mathcal{X}_i={\mathbb R}^m$); public constraints ($\mathcal{X}_i=\mathcal{X}$);
 and private constraints $\mathcal{X}_i$. \emph{The underlying communication network} can either
 be static (i.e., not varying in time,) or dynamic (i.e., varying in time.)
  A dynamic network can be deterministically or randomly varying. \emph{Link failures}
 can be symmetric or asymmetric;
 that is, the random network realizations can be symmetric or, more generally,
  asymmetric graphs, the latter case being more challenging in general. \emph{The communication protocol} can either be synchronous, or asynchronous, i.e., of gossip~type,~\cite{BoydGossip}. We next review the existing work with respect to these four dimensions.

    \mypar{Primal subgradient algorithms}
    References~\cite{asu-new-jbg,asu-random,nedic_T-AC,nedic_novo}
    and~\cite{nedic-gossip} develop primal subgradient
      algorithms, with~\cite{asu-new-jbg,asu-random,nedic_T-AC,nedic_novo} assuming synchronous communication. References~\cite{nedic_T-AC}~and~\cite{nedic_novo}
       consider a deterministically varying network, with~\cite{nedic_T-AC}
       for the unconstrained
        problem and~\cite{nedic_novo} for public constraints. References~\cite{asu-random} and~\cite{asu-new-jbg}
        consider random networks; reference~\cite{asu-random}
         is for public constraints, while~\cite{asu-new-jbg} assumes private constraints.
           Both references~\cite{asu-new-jbg}~and~\cite{asu-random} essentially handle only \emph{symmetric}
            link failures, namely, they use
           local weighted averaging as an intermediate step
            in the update rule and constrain the corresponding averaging matrix to be \emph{doubly} stochastic.
             In practice, these translate into requiring symmetric graph realizations,
             and, consequently, symmetric link failures. Reference~\cite{nedic-gossip}
             presents a primal subgradient algorithm for unconstrained optimization and static network and uses the gossip communication protocol.
             Finally, reference~\cite{nedic-regression} studies a generalization of problem~\eqref{eqn_original_primal_problem},
             where the objective function $\sum_{i=1}^N f_i(x)$ is replaced by $g \left( \sum_{i=1}^N f_i(x) \right)$;
             that is, an outer public, convex function $g(\cdot)$ is introduced. The optimization problem in~\cite{nedic-regression} has public constraints, the communication is synchronous, and the network is deterministically time varying. Reference~\cite{nedic-regression} proposes a distributed algorithm where each node $i$, at each time step $k$, updates
              two quantities: an estimate of the optimal solution $x_i(k)$, and an estimate of the quantity
              $(1/N)\sum_{j=1}^N f_j(x_i(k))$, by communicating with immediate neighbors only. When the algorithm in~\cite{nedic-regression} is applied to~\eqref{eqn_original_primal_problem}, it reduces to the primal subgradient~in~\cite{nedic_novo}.

  \mypar{Primal-dual algorithms} As far as we are aware, primal-dual algorithms have been studied only for
  \emph{static networks}. For example,
    references~\cite{bazerque_lasso,bazerque_sensing} consider a special case of~\eqref{eqn_original_primal_problem},
     namely, the Lasso (least-absolute shrinkage and selection operator) type problem.
      They propose the AD-MoM (alternating direction method of multipliers) type primal-dual
      algorithms for \emph{static networks}, \emph{synchronous communication}, and \emph{no constraints.}            Reference~\cite{BoydADMoM} applies AD-MoM to various statistical learning problems, including Lasso, support vector machines, and sparse logistic regression, assuming a parallel network architecture
        (all nodes communicate with a fusion node,) synchronous communication, and no~link~failures.

      In this paper,
      rather than subgradient type,
      we provide and develop a AL primal-dual
       algorithm for the
       optimization~\eqref{eqn_original_primal_problem}
        with private costs and private constraints,
        random networks, and asynchronous gossip communication.
         In contrast with existing work on primal-dual methods,
         for example,~\cite{bazerque_lasso,bazerque_sensing}, our \textbf{AL--G} handles \emph{private constraints}, \emph{random networks}, \emph{asymmetric link failures},
           and \emph{gossip communication.}\footnote{\textbf{AL--G} algorithm uses
           asynchronous gossip communication, but it is not completely asynchronous algorithm, as it updates
           the dual variables synchronously, at a slow time scale (as details in Section~{IV}.)}

%
%
%

\mypar{Paper organization} Section {II} introduces the communication and computational model. Section {III} presents the \textbf{AL--G} algorithm for the networks with link failures.
Section {IV} proves the convergence of the \textbf{AL--G} algorithm. Section {V} studies the variants to \textbf{AL--G},
the \textbf{AL--MG}, and \textbf{AL--BG} algorithms. Section {VI} provides two simulation examples: 1) $l_1$--regularized logistic regression for classification; and 2) cooperative spectrum sensing for cognitive radios. Finally, section {VII} concludes the paper.
The Appendix proves convergence of \textbf{AL--MG} and \textbf{AL--BG}.

\vspace{-3mm}
\section{Problem model}
\label{section_Problem_Model}
This section explains the communication model (the time slotting, the communication protocol, and the link failures,)
 and the computation model (assumptions underlying the optimization problem~\eqref{eqn_original_primal_problem}.)

\mypar{Network model: Supergraph} The connectivity of the networked system is described by the bidirectional,
   connected supergraph $G = \left( \mathcal{N}, E \right)$, where $\mathcal{N}$ is the set of nodes (with cardinality
 $|\mathcal{N}|=N$) and $E$ is the set of bidirectional edges $\{i,j\}$ ($|E|=M$).
 The supergraph $G$ is simple, i.e., there are no self-edges. Denote by $\Omega_i \subset \mathcal{N}$, the neighborhood set of node $i$ in $G$, with cardinality $d_i=|\Omega_i|$. The integer $d_i$ is the (supergraph) degree of node $i$. The supergraph $G$ models and collects all (possibly unreliable) communication channels in the network; actual network realizations during the algorithm run will be \emph{directed} subgraphs of $G$. We denote the directed edge (arc) that originates in node $i$ and ends in node $j$ either by $(i,j)$ or $i \rightarrow j$, as appropriate.
 The set of all arcs is:
$
E_d = \{ (i,j): \{i,j\} \in E  \},
$ where $|E_d|=2M$. We assume that the supergraph is known, i.e., each node knows a priori
with whom it can communicate (over a possibly unreliable link.)
\vspace{-3mm}

\mypar{Optimization model} We summarize the assumptions on the cost functions $f_i(\cdot)$ and $f(\cdot)$, $f(x):=\sum_{i=1}^N f_i(x)$, and the constraint sets
$\mathcal{X}_i$ in~(1):
\begin{assumption}
We assume the following for the optimization problem~\eqref{eqn_original_primal_problem}:
\begin{enumerate}
\item The functions $f_i:\,
{\mathbb R}^m \rightarrow \mathbb R$ are convex and coercive, i.e., $f_i(x) \rightarrow \infty$ whenever $\|x\| \rightarrow \infty$.
\item The constraint sets $\mathcal{X}_i \subset {\mathbb R}^m$ are closed and convex, and
$\mathcal{X}:=\cap_{i=1}^N \mathcal{X}_i$ is nonempty.
\item \textbf{(Regularity condition)} There exists a point $x_0 \in \mathrm{ri} \left( \mathcal{X}_i\right)$, for all $i=1,...,N$.
\end{enumerate}
\end{assumption}
Here $\mathrm{ri}\left( \mathcal{S}\right)$ denotes the relative
interior of a set $\mathcal{S} \subset {\mathbb R}^m$ (see~\cite{Urruty}).
%
We will derive the \textbf{AL--G} algorithm to solve~\eqref{eqn_original_primal_problem} by first reformulating it (see ahead eqn.~(2),) and then
 dualizing the reformulated problem (using AL dual.) Assumption 1.3 will play a role to assure strong duality. This will be detailed in subsection {III-A}. Note that Assumption 1.3 is rather mild, saying only that
 the intersection of the $\mathcal{X}_i$'s, $i=1,...,N$, is ``large" enough to contain a
 point from the relative interior of each of the $\mathcal{X}_i$'s. Denote by $f^\star$ the optimal value and $\mathcal{X}^\star = \left\{ x^\star \in \mathcal{X}:\, \sum_{i=1}^N f_i(x^\star)=f^\star\right\}$ the solution set to~\eqref{eqn_original_primal_problem}. Under Assumptions~1, $f^\star$ is finite, and $\mathcal{X}^\star$ is nonempty, compact, and convex,~\cite{NedicNotes}.
   The model~\eqref{eqn_original_primal_problem} applies also when $\mathcal{X}_i = {\mathbb R}^m$,
 for $i$'s in a subset of $\{1,...,N\}$. The functions $f_i(\cdot)$, $f(\cdot)$
  need not be differentiable; $f(\cdot)$ satisfies
   an additional mild assumption detailed in~Section~{IV}.

We now reformulate~(1) to derive the \textbf{AL--G} algorithm. Start by cloning the variable $x \in {\mathbb R}^m$ and attaching a local copy of it, $x_i \in {\mathbb R}^m$, to each node in the network. In addition, introduce the variables $y_{ij} \in {\mathbb R}^m $ and $y_{ji} \in {\mathbb R}^m $, attached to each link $\{i,j\}$ in the supergraph. To keep the reformulated problem equivalent to~\eqref{eqn_original_primal_problem}, we introduce coupling constraints $x_i = y_{ij},\,\, (i,j) \in E_d$ and $y_{ij} = y_{ji},\,\, \{i,j\} \in E$. The reformulated optimization problem becomes:
\begin{equation}
\begin{array}[+]{ll}
\mbox{minimize} & \sum_{i=1}^N f_i(x_i) \\
\mbox{subject to} & x_i \in \mathcal{X}_i,\,\,i=1,...,N,
\\ & x_i = y_{ij},\,\, (i,j) \in E_d \\
& y_{ij} = y_{ji},\,\, \{i,j\} \in E.
\end{array}
\label{eqn_primal_reformulated_with_y_s}
\end{equation}
\begin{figure}[thpb]
      \centering
      \includegraphics[trim = 10mm 85.mm 10mm 88.mm, clip, width=10cm]{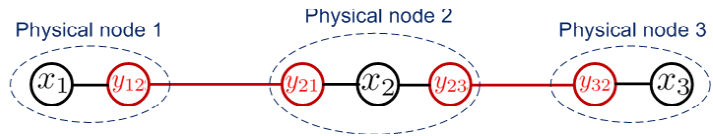}
      \caption{ Illustration of the reformulation~\eqref{eqn_primal_reformulated_with_y_s} for a chain supergraph with $N=3$ (physical) nodes. }
      \label{figure-proof}
\end{figure}
The variables $x_i$ and $y_{ij}$ may be interpreted as virtual nodes in the network (see Figure~1.) Physically,
the variables $x_i$, $y_{ij}$, $j \in \Omega_i$ are maintained by (physical) node $i$. The virtual link between
 nodes $x_i$ and $y_{ij}$ is reliable (non-failing,) as both $x_i$ and $y_{ij}$ are physically maintained by node $i$. On the other
 hand, the virtual link between $y_{ij}$ and $y_{ji}$ may be unreliable (failing,) as this link corresponds
 to the physical link between nodes $i$ and $j$.

The optimization problems~\eqref{eqn_original_primal_problem} and~\eqref{eqn_primal_reformulated_with_y_s}
 are equivalent because the supergraph is connected. The optimal value for~\eqref{eqn_primal_reformulated_with_y_s} is
 equal to the optimal value for~\eqref{eqn_original_primal_problem} and equals $f^\star$; the set of solutions to~\eqref{eqn_primal_reformulated_with_y_s}
  is $\left\{ \{x_i^\star\}, \{y_{ij}^\star\}:\,x_i^\star=x^\star, \forall i=1,...,N,\,y_{ij}^\star=x^\star,\,\forall (i,j) \in E_d,\,
  \mathrm{for\,\,some\,}x^\star \in \mathcal{X}^\star \right\}$.

%
%
%
%

\mypar{Time slotting} As we will see in section~{III}, the \textbf{AL--G} algorithm (and also its
 variants \textbf{AL--MG} and \textbf{AL--BG} in section V) is based on the
AL dual of~\eqref{eqn_primal_reformulated_with_y_s}. The \textbf{AL--G} operates at 2 time scales:
 the dual variables are updated at a slow time scale, and the primal variables are updated at a fast time scale. Thus,
accordingly, the time is slotted with: 1) slow time scale slots $\{t\}$; and 2) fast time scale slots $\{k\}$. Fast time scale slots (for the primal variables update) involve asynchronous communication between the nodes in the network and
are detailed in the next paragraph. At the end of each $t$-slot, there is an idle time interval with no
communication, when the dual variables are updated. The dual variables update at each node requires no communication.

\mypar{Fast time scale slots $\{k\}$ and asynchronous communication model}
We now define the fast time scale slots $\{k\}$ for the asynchronous communication
 and the primal variables update. We assume the standard model for asynchronous communication~\cite{BoydGossip,Scaglione}.
Each node (both physical and virtual) has a clock that ticks (independently across nodes)
 according to a $\lambda$-rate Poisson process. Denote the clocks
 of $x_i$ and $y_{ij}$ by $T^x_i$ and $T^y_{ij}$, respectively. If $T^x_i$ ticks,
 a virtual communication from $y_{ij}$, $\forall j \in \Omega_i$, to $x_i$, follows. With the \textbf{AL--G}
  algorithm, this will physically correspond to the update of the variable $x_i$, as we will see later. If the clock $T^y_{ij}$
   ticks, then (virtual) node $y_{ij}$ transmits to $y_{ji}$ (physically, node $i$ transmits to node $j$.)
    We will see later that, after a (successful) communication $y_{ij}\rightarrow y_{ji}$,
    the update of $y_{ji}$ follows. We also introduce a virtual clock $T$ that ticks whenever one of the clocks $T^x_i$,
 $T_{ij}^y$, ticks; the clock $T$ ticks according to a $(N+2M)$--rate Poisson process. Denote by $\tau_k$, $k=1,2,...$ the times when the $k$-th tick of $T$ occurs. The time is slotted
 and the $k$-th slot is $[\tau_{k-1}, \tau_{k})$, $\tau_0=0$, $k=1,2,...$\footnote{For notation simplicity,
  at the beginning of each $t$--slot, we reset $\tau_0$ to zero, and we start counting the $k$--slots
   from $k=1$.}
\vspace{-3mm}

 \mypar{Random link failures}
  Motivated by applications in wireless networked systems, we allow that transmissions $y_{ij} \rightarrow y_{ji}$ may fail. (Of course, the transmissions through the virtual links
  $y_{ij}\rightarrow x_i$ do not fail.) To formally account for
   link failures, we define the $N\times N$ random adjacency
    matrices $A(k)$, $k=1,2,...$; the matrix $A(k)$
     defines the set of available physical links at time slot $k$.
     We assume that the link failures are
      temporally independent, i.e., $\{A(k)\}$ are independent identically distributed (i.i.d.)
      The entries $A_{ij}(k)$, $(i,j) \in E_d$, are Bernoulli random variables, $A_{ij}(k) \sim
      \mathrm{Bernoulli}(\pi_{ij})$, $\pi_{ij}=\mathrm{Prob}\left( A_{ij}(k)=1\right)>0$, and $A_{ij}(k) \equiv 0$, for $(i,j) \notin E_d$. We allow $A_{ij}(k)$ and $A_{lm}(k)$ to be correlated.\footnote{With \textbf{AL--MG} algorithm, in Section VI, we will additionally require $A_{ij}(k)$ and $A_{lm}(k)$ be independent.}
       At time slot $k$,
        at most one link $(i,j)\in E_d$ is activated for transmission. If it is available at time $k$,
         i.e., if $A_{ij}(k)=1$, then the transmission
          is successful; if the link $(i,j)$ is unavailable ($A_{ij}(k)=0$,) then the transmission is unsuccessful. We assume naturally that the Poisson process
          that governs the ticks of $T$ and the adjacency matrices $A(k)$, $k=1,2,...$ are independent. Introduce the ordering of
           links $(i,j) \in E_d$, by attaching a distinct number $l$, $l=1,...,2M$, to each link $(i,j)$; symbolically, we write
            this as $l \sim (i,j)$. Introduce now the random variables $\zeta(k)$, $k=1,2,...$, defined as follows:
1) $\zeta(k)=i$, if the $k$-th tick of $T$ comes from $T_i^x$; 2) $\zeta(k)=N+l$, $l \sim (i,j)$, if the $k$-th tick of $T$ comes from $T_{ij}^y$ and $A_{ij}(k)=1$;
and 3) $\zeta(k)=0$, otherwise. It can be shown that $\zeta(k)$, $k=1,2,...$, are i.i.d. The random variables $\zeta(k)$
 define the order of events in our communication model. For
 example, $\zeta\eqref{eqn_original_primal_problem}=N+l$, $l \sim(i,j)$, means that, at time slot $k=$1, the virtual node $y_{ij}$ successfully
  transmitted data to the virtual node $y_{ji}$. We
  remark that $\mathrm{Prob} \left( \zeta(k)=s \right)$
   is strictly positive, $\forall s=0,1,...,N+2M$. This fact will be important when studying the convergence~of~\textbf{AL--G}.

          The communication model in this paper, with \emph{static} supergraph and link failures,
          is standard for networked systems supported by wireless communication and static (non moving) nodes, see, e.g., \cite{SoummyaTopologyDesign,BroadcastGossipFailures}. The model needs to be modified for
           scenarios with moving nodes (e.g., mobile robots) where the supergraph itself can be time varying. This
           is not considered here.

\vspace{-3mm}
\section{\textbf{AL--G} algorithm (augmented Lagrangian gossiping)}
\vspace{-3mm}

%
This section details the \textbf{AL--G} algorithm for solving \eqref{eqn_original_primal_problem}.
In subsection~\ref{subsect_dualization}, we dualize \eqref{eqn_primal_reformulated_with_y_s} to form the AL dual of problem~\eqref{eqn_primal_reformulated_with_y_s}. Subsection IV-B details
the D--AL--G algorithm for the dual variable update, at a slow time scale; subsection IV-C
details P--AL--G to update the primal variables, at a fast~time~scale.
\vspace{-3mm}
\subsection{Dualization}
\vspace{-3mm}
\label{subsect_dualization}
We form the AL dual of the optimization problem~\eqref{eqn_primal_reformulated_with_y_s} by dualizing all the constraints of the type $x_i = y_{ij}$ and $y_{ij}=y_{ji}$. The dual variable that corresponds to the constraint $x_i=y_{ij}$ will be denoted by $\mu_{(i,j)}$, the dual variable that corresponds to the (different) constraint $x_j=y_{ji}$ will be denoted by $\mu_{(j,i)}$, and the one that corresponds to $y_{ij}=y_{ji}$ is denoted by $\lambda_{\{i,j\}}$. In the algorithm implementation, both nodes $i$ and $j$ will maintain their own copy of
the variable $\lambda_{\{i,j\}}$--the variable $\lambda_{(i,j)}$ at node $i$ and the variable $\lambda_{(j,i)}$ at node $j$.
Formally, we use both $\lambda_{(i,j)}$ and $\lambda_{(j,i)}$, and we add the constraint
 $\lambda_{(i,j)}=\lambda_{(j,i)}$. The term after dualizing $y_{ij}=y_{ji}$, equal to $\lambda_{\{i,j\}}^\top (y_{ij}-
  y_{ji})$, becomes: $\lambda_{(i,j)}^\top y_{ij}-
 \lambda_{(j,i)}^\top y_{ji}$. The resulting AL dual function $L_a(\cdot)$, the (augmented) Lagrangian $L(\cdot)$, and the AL dual optimization problem are, respectively,
  given in eqns.~\eqref{eqn-dual-fcn},~\eqref{eqn-lagrangian},~and~\eqref{eqn-dual-problem-al}.
\begin{eqnarray}
\label{eqn-dual-fcn}
\begin{array}[+]{ll}
L_a \left( \, \{ \lambda_{(i,j)} \}, \{ \mu_{(i,j)} \}  \, \right)  \,= \,
\,\,\, \mbox{min} & L \left(  \{x_i\},\, \{y_{ij}\},\, \{  \lambda_{(i,j)} \}, \{\mu_{(i,j)}\}  \right) \,\,\,\,\,\,\,\,\,\,\,\\
\,\,\,\,\,\,\,\,\,\,\,\,\,\,\,\,\,\,\,\,\,\,\,\,\,\,\,\,\,\,\,\,\,\,\,\,\,\,\,
\,
\,\,\,\,\,\,\,\,\,\,\,\,\,\,\,\,\,\,\,\,\,\,\,\,\,\,\,\,
\mbox{subject to} & x_i \in \mathcal{X}_i,\, i=1,...,N   \,\,\,\,\,\,\,\,\,\,\,\\
                  & y_{ij} \in {\mathbb R}^m,\,(i,j) \in E_d
\end{array}
\end{eqnarray}
\begin{eqnarray}
\label{eqn-lagrangian}
L \left(  \{x_i\},\{ y_{ij} \}, \, \{\lambda_{(i,j)}\}, \{\mu_{(i,j)} \}  \right) \,= \,
 \sum_{i=1}^N f_i(x_i) + \sum_{(i,j) \in E_d} \mu_{(i,j)}^{\top}\, \left( x_i - y_{ij} \right) \,\,\,\,\,\,\,\,\,\,\,\,\,\,\,\, \,\,\, \\
 +   \sum_{\{i,j\} \in E, \,i<j}  \lambda_{(i,j)}^{\top}\,  y_{ij} -  \lambda_{(j,i)}^{\top} \, y_{ji}
 + \frac{1}{2} \rho \sum_{(i,j) \in E_d}\,\| x_i - y_{ij} \|^2 +
  \frac{1}{2} \rho \sum_{\{i,j\} \in E,\,i<j}\,\| y_{ij} - y_{ji} \|^2  \nonumber
\end{eqnarray}
\begin{equation}
\label{eqn-dual-problem-al}
\begin{array}[+]{ll}
\mbox{maximize} & L_a \left(  \{\lambda_{(i,j)}\}  , \, \{ \mu_{(i,j)} \} \right) \\
\mbox{subject to} & \lambda_{(i,j)}=  \lambda_{(j,i)},\,\,\{i,j\} \in E\\
& \mu_{(i,j)} \in {\mathbb R}^m,\,\,(i,j) \in E_d
\end{array}.
\end{equation}
In eqn.~\eqref{eqn-lagrangian}, $\rho$ is a positive parameter. See~\cite{BertsekasOptimization} for some background
on AL methods. The terms $ \lambda_{(i,j)}^{\top}\,  y_{ij} -  \lambda_{(j,i)}^{\top} \, y_{ji} $ in the sum $\sum_{\{i,j\} \in E} {  \lambda_{(i,j)}^{\top}\,  y_{ij} -  \lambda_{(j,i)}^{\top} \, y_{ji}    }$ are arranged such that $i<j$, for all
$\{i,j\} \in E$.
\footnote{For each link $\{i,j\} \in E$, the virtual nodes $y_{ij}$ and $y_{ji}$ (i.e., nodes $i$ and $j$,) have to agree
 beforehand (in the network training period) which one takes the $+$ sign and which one takes the $-$ sign in $ \lambda_{(i,j)}^{\top}\,  y_{ij} -  \lambda_{(j,i)}^{\top} \, y_{ji} $. In eqn.~\eqref{eqn-lagrangian}, for sake of notation simplicity, the distribution of $+$ and $-$ signs at each link $\{i,j\}$ is realized by the order of node numbers, where a distinct number in $\{1,...,N\}$ is assigned to each node. However, what matters is only to assign
 $+$ to one node (say $i$) and $-$ to the other, for each $\{i,j\} \in E$.}

Denote by $d^\star$ the optimal value of the dual problem~\eqref{eqn-dual-problem-al}, the dual of~\eqref{eqn_primal_reformulated_with_y_s}. Under Assumption~1,
the strong duality between~\eqref{eqn_primal_reformulated_with_y_s} and~\eqref{eqn-dual-problem-al} holds, and $d^\star=f^\star$; moreover, the set
 of optimal solutions $\mathcal{D}^\star = \left\{ \{\lambda^\star_{(i,j)}\}  , \, \{ \mu^\star_{(i,j)}\}:
 \, L_a \left(  \{\lambda^\star_{(i,j)}\}  , \, \{ \mu^\star_{(i,j)} \} \right)=f^\star \right\}$ is nonempty.
  Denote by $C:=\mathcal{X}_1 \times \mathcal{X}_2 \times ... \times \mathcal{X}_N \times \left({\mathbb R}\right)^{m(2M)}$
   the constraint set in~\eqref{eqn-dual-fcn}, i.e., the constraints in~\eqref{eqn_primal_reformulated_with_y_s} that are not dualized. Let $x_0$ be a point in $\mathrm{ri}(\mathcal{X}_i)$, $i=1,...,N$ (see Assumption 1.3.) Then, a point $\left( \{x_{i,0}\}, \{y_{ij,0}\}\right) \in C$, where $x_{i,0}=x_0$, $y_{ij,0}=x_0$, belongs to $\mathrm{ri}(C)$,
    and it clearly satisfies all equality constraints in the primal problem~\eqref{eqn_primal_reformulated_with_y_s}; hence, it is a Slater point, and the
    above claims on strong duality hold, \cite{Urruty}. We remark that
     strong duality holds for any choice of $\rho \geq 0$ (but we are interested
     only in the case $\rho>0$,) and, moreover, the set of dual solutions $\mathcal{D}^\star$
     does not depend on the choice of $\rho$, provided that $\rho \geq 0$ (see, e.g., \cite{Rockafellar}, p.359.)
%
%
%
%
%
\vspace{-3mm}
\subsection{Solving the dual: D--AL--G (dual augmented Lagrangian gossiping) algorithm}
We now explain how to solve the dual problem~\eqref{eqn-dual-problem-al}. First, we note that~\eqref{eqn-dual-problem-al} is equivalent
to the unconstrained maximization of $L_a^{\prime} \left( \{\lambda_{ \{i,j\}}\}, \{\mu_{(i,j)}\} \right)
=\min_{x_i \in \mathcal{X}_i, y_{ij}\in {\mathbb R}^m} L^\prime
\left(
  \{x_i\},\{ y_{ij} \}, \, \{\lambda_{\{i,j\}}\}, \{\mu_{(i,j)} \}  \right)
$, where the function $L^\prime
\left(
 \{x_i\},\{ y_{ij} \}, \, \{\lambda_{\{i,j\}}\}, \{\mu_{(i,j)} \}  \right)$
 is defined by replacing both $\lambda_{(i,j)}$ and $\lambda_{(j,i)}$
  in $L(\cdot)$ (eqn.~\eqref{eqn-lagrangian}) with $\lambda_{\{i,j\}}$, for all $\{i,j\} \in E$.
 The standard method of multipliers for the unconstrained maximization of $L_a^{\prime}(\cdot)$
  is given by:
\begin{eqnarray}
\label{eqn_dual_update_rule_lambda_prel}
\lambda_{\{i,j\}}(t+1) &=& \lambda_{\{i,j\}}(t) + \rho \,\, \mathrm{sign}(j-i)\, \left(  y_{ij}^{\star}(t) - y_{ji}^{\star}(t)  \right)
,\,\{i,j\}\in E     \\
\mu_{(i,j)}(t+1) &=& \mu_{(i,j)}(t) + \rho \, \left(  x_i^{\star}(t) - y_{ij}^{\star}(t)  \right), (i,j)\in E_d. \nonumber
\end{eqnarray}
\vspace{-2mm}
\begin{eqnarray}
\label{eqn_solve_for_x_i_y_ij_prel}
\begin{array}[+]{ll}
\left( \, \{x_i^{\star}(t)\} ,\, \{y_{ij}^{\star}(t)\}  \, \right)\, \in \,
\mbox{arg min} & L^{\prime} \left(  \{x_i\},\,\{y_{ij}\},\, \{\lambda_{\{i,j\}}(t)\},\,\{\mu_{(i,j)}(t)\}  \right) \\
\,\,\,\,\,\,\,\,\,\,\,\,\,\,\,\,\,\,\,\,\,\,\,\,\,\,\,\,\,\,\,\,\,\,\,\,\,\,\,
\,\,\,\,\,\,\,\,\,\,\,\,\,\,\,\,\,\,\,\,
\mbox{subject to} & x_i \in \mathcal{X}_i,\, i=1,...,N \\
                  & y_{ij} \in {\mathbb R}^m,\,\, (i,j) \in E_d.
\end{array}
\end{eqnarray}
Assigning a copy of $\lambda_{\{i,j\}}$ to both nodes $i$ (the corresponding copy is $\lambda_{(i,j)}$) and $j$ (the
corresponding copy is $\lambda_{(j,i)}$), eqn.~\eqref{eqn_dual_update_rule_lambda_prel} immediately yields an algorithm to solve \eqref{eqn-dual-problem-al}, given by:
%
%
%
\begin{eqnarray}
\label{eqn_dual_update_rule_lambda_mu}
\lambda_{(i,j)}(t+1) &=& \lambda_{(i,j)}(t) + \rho \,\, \mathrm{sign}(j-i)\, \left(  y_{ij}^{\star}(t) - y_{ji}^{\star}(t)  \right)     , \, (i,j) \in E_d\\
\mu_{(i,j)}(t+1) &=& \mu_{(i,j)}(t) + \rho \, \left(  x_i^{\star}(t) - y_{ij}^{\star}(t)  \right), \, (i,j) \in E_d, \nonumber
\end{eqnarray}
where
\begin{eqnarray}
\label{eqn_solve_for_x_i_y_ij}
\begin{array}[+]{ll}
\left( \, \{x_i^{\star}(t)\} ,\, \{y_{ij}^{\star}(t)\}  \, \right)\, \in \,
\mbox{arg min} & L \left(  \{x_i\},\,\{y_{ij}\},\, \{\lambda_{(i,j)}(t)\},\,\{\mu_{(i,j)}(t)\}  \right) \\
\,\,\,\,\,\,\,\,\,\,\,\,\,\,\,\,\,\,\,\,\,\,\,\,\,\,\,\,\,\,\,\,\,\,\,\,\,\,\,
\,\,\,\,\,\,\,\,\,\,\,\,\,\,\,\,\,\,\,\,
\mbox{subject to} & x_i \in \mathcal{X}_i,\, i=1,...,N \\
                  & y_{ij} \in {\mathbb R}^m,\,\, (i,j) \in E_d.
\end{array}
\end{eqnarray}
(Note that $\left( \, \{x_i^{\star}(t)\} ,\, \{y_{ij}^{\star}(t)\}  \, \right)$ is the same in \eqref{eqn_solve_for_x_i_y_ij_prel} and \eqref{eqn_solve_for_x_i_y_ij}.) According to eqn.~\eqref{eqn_dual_update_rule_lambda_mu}, essentially,
 both nodes $i$ and $j$ maintain their own copy ($\lambda_{(i,j)}$ and $\lambda_{(j,i)}$, respectively) of the same variable, $\lambda_{\{i,j\}}$.
%
%
%
%
It can be shown that, under Assumption~1, any limit point of the sequence $\left( \{x_i^\star(t)\}, \{y_{ij}^\star(t)\}\right)$, $t=0,1,...$, is a solution of~\eqref{eqn_primal_reformulated_with_y_s} (see, e.g.,~\cite{BertsekasBook}, Section 3.4);
 and the corresponding limit point of the sequence $x_i^\star(t)$, $t=0,1,...$, is a solution of~\eqref{eqn_original_primal_problem}.

Before updating the dual variables as in~\eqref{eqn_dual_update_rule_lambda_mu}, the nodes need to solve problem~\eqref{eqn_solve_for_x_i_y_ij},
with fixed dual variables, to get
$\left( \{x_i^\star(t)\}, \{y_{ij}^\star(t)\}\right)$. We will explain in the next subsection~(IV-C),
how the P--AL--G algorithm solves problem~\eqref{eqn_solve_for_x_i_y_ij}
in a distributed, iterative way, at a fast time scale $\{k\}$. We remark that P--AL--G
 terminates after a finite number of iterations $k$, and thus produces an inexact solution
  of~\eqref{eqn_solve_for_x_i_y_ij}. We will see that, after termination of the P--AL--G algorithm, an inexact solution for
  $y_{ji}$ is available at node $i$; denote it by $y_{ji}^L(t)$. Denote, respectively,
  by $x_i^F(t)$
   and $y_{ij}^F(t)$,
  the inexact solutions for $x_i$ and $y_{ij}$ at node $i$, after termination of P--AL--G. Then, the implementable update of the dual variables~is:
\begin{eqnarray}
\label{eqn_dual_INEXACT_update_rule_lambda_mu}
\lambda_{(i,j)}(t+1) &=& \lambda_{(i,j)}(t) + \rho \,\, \mathrm{sign}(j-i)\, \left(  y_{ij}^{F}(t) - y_{ji}^{L}(t)  \right)     \\
\mu_{(i,j)}(t+1) &=& \mu_{(i,j)}(t) + \rho \, \left(  x_i^{F}(t) - y_{ij}^{F}(t)  \right). \nonumber
\end{eqnarray}
Note that the ``inexact'' algorithm in~\eqref{eqn_dual_INEXACT_update_rule_lambda_mu} differs
from \eqref{eqn_dual_update_rule_lambda_mu} in that it does not guarantee that $\lambda_{(i,j)}(t)=\lambda_{(j,i)}(t)$, due to a finite time termination
 of P--AL--G.
\vspace{-3mm}
\subsection{Solving for~\eqref{eqn_solve_for_x_i_y_ij}: P--AL--G algorithm}
Given $\{\lambda_{(i,j)}(t)\},\,\{\mu_{(i,j)}(t)\}$, we solve the primal problem~\eqref{eqn_solve_for_x_i_y_ij} by a randomized, block-coordinate, iterative algorithm, that we refer to as P--AL--G. To simplify notation, we will write only $ \lambda_{ (i,j)} $
and $\mu_{(i,j)}$ instead of $\lambda_{(i,j)}(t)$, $\mu_{(i,j)}(t)$. We remark that $\lambda_{(i,j)}(t)$, $\mu_{(i,j)}(t)$ stay fixed while the optimization in eqn.~\eqref{eqn_solve_for_x_i_y_ij} is done (with respect to $x_i$, $y_{ij}$.)

%
%
%
The block-coordinate iterative algorithm works as follows: at time slot $k$, the function in~\eqref{eqn-lagrangian} is optimized with respect to a single block-coordinate, either $x_i$ or $y_{ij}$, while other blocks are fixed. Such an algorithm for solving~\eqref{eqn_solve_for_x_i_y_ij} admits distributed implementation, as we show next. Minimization of the function $L \left(  \{x_i\},\{ y_{ij} \}, \, \{\lambda_{(i,j)}\}, \{\mu_{(i,j)} \}  \right)$ with respect to $x_i$, while the other coordinates $x_j$ and $y_{ij}$ are fixed, is equivalent to the following problem:
\begin{equation}
 \begin{array}[+]{ll}
\mbox{minimize} & f_i(x_i) + \left(    \overline{\mu}_i - \rho\, \overline{y}_{i}     \right)^{\top} x_i
 + \frac{1}{2} \rho \,d_i\|x_i\|^2\\
\mbox{subject to} & x_i \in \mathcal{X}_i
\end{array},
\label{eqn_solving_wrt_x_i-simplified}
 \end{equation}
where $\overline{\mu}_i=\sum_{j \in \Omega_i} \mu_{(i,j)}$ and $\overline{y}_i = \sum_{j \in \Omega_i}y_{ij}$. Thus, in order to update $x_i$, the node $i$
 needs only information from its (virtual) neighbors. Minimization of the function
 $L \left(  \{x_i\},\{ y_{ij} \}, \,
  \{ \lambda_{(i,j)}\}, \{\mu_{(i,j)} \}  \right)$ with respect to $y_{ij}$, while the other coordinates $x_j$ and $y_{lm}$ are fixed, is equivalent to:
\begin{equation}
 \begin{array}[+]{ll}
\mbox{minimize} &  \mu_{(i,j)}^{\top}\, \left( x_i - y_{ij} \right) + \lambda_{(i,j)}^{\top}
\mathrm{sign}(j-i)
\, \left( y_{ij} - y_{ji} \right)
 + \frac{1}{2} \rho \| x_i - y_{ij} \|^2  +   \frac{1}{2} \rho \| y_{ij} - y_{ji} \|^2\\
\mbox{subject to} & y_{ij} \in {\mathbb R}^m.
\end{array}
\label{solve_wrt_y_ij}
 \end{equation}
Thus, in order to update $y_{ij}$, the corresponding virtual node needs only to communicate information with its neighbors in the network, $y_{ji}$ and $x_i$. Physical communication is required only with $y_{ji}$ (i.e., with physical node $j$.) The optimization problem~\eqref{solve_wrt_y_ij} is an unconstrained problem with convex quadratic cost and admits the closed form solution:
\begin{equation}
\label{eqn_solution_y_ij}
y_{ij} = \frac{1}{2}y_{ji}+\frac{1}{2}x_i+\frac{1}{2 \rho} \left( \mu_{(i,j)} - \mathrm{sign}(j-i)\,\lambda_{(i,j)}  \right) .
\end{equation}
%
%
%
%
%
%
%
%
%
\mypar{Distributed implementation}
We have seen that the block-coordinate updates in eqns.~\eqref{eqn_solving_wrt_x_i-simplified}~and~\eqref{eqn_solution_y_ij} only
 require neighborhood information at each node. We next give the distributed implementation
  of P--AL--G (see Algorithm~1) using the asynchronous communication model defined in section~{II}.
\vspace{-5mm}
\begin{algorithm}
\caption{  Algorithm with gossiping for solving~\eqref{eqn_solve_for_x_i_y_ij} (P--AL--G)  }
\begin{algorithmic}[1]
    \REPEAT
        \STATE Wait for the tick of one of the clocks $T_j^x$, $T_{ij}^y$.
        \STATE If clock $T_{ij}^y$ ticks, node $i$ transmits to node $j$ the current value of $y_{ij}$.
        \\ If node $j$ successfully receives $y_{ij}$, it updates the variable $y_{ji}$ according to the equation~\eqref{eqn_solution_y_ij}.
          \STATE If clock $T_i^x$ ticks, node $i$ updates the variable $x_i$ by solving~\eqref{eqn_solving_wrt_x_i-simplified}.
    \UNTIL a stopping criterion is met.
\end{algorithmic}
\end{algorithm}
\vspace{-2mm}


\subsubsection{Simplified notation and an abstract model of the P--AL--G}
\label{subsubsec-abst-model-of-p-sng}
We now simplify the notation and introduce an abstract model for the P--AL--G algorithm, for
 the purpose of convergence analysis in Section~{IV}.
  Denote, in a unified way, by $z_i$, the primal variables $x_i$ and $y_{ij}$,
   i.e., $z_i:=x_i$, $i=1,...,N$, and $z_l:=y_{ij}$, $l=N+1,...,N+2M$, $(i,j) \in E_d$.
   Then, we can write the function in~\eqref{eqn-lagrangian}, viewed as a function of
   the primal variables, simply as $L(z)$, $L:{\mathbb R}^{m(N+2M)}\rightarrow \mathbb R$.
   Also, denote in a unified way the
   constraint sets $C_i:=\mathcal{X}_i$, $i=1,...,N$,
    and $C_l:={\mathbb R}^m$, $l=N+1,...,2M+N$ ($C_l$, $l=N+1,...,N+2M$;
     these sets correspond to the constraints
     on $y_{ij}$, $(i,j)\in E_d$.) Finally, define $C:=C_1 \times C_2\times ...\times C_{N+2M}$.
     Thus, the optimizations in~\eqref{eqn_solving_wrt_x_i-simplified}~and~\eqref{eqn_solution_y_ij} are simply minimizations of $L(z)$
     with respect to a single (block) coordinate $z_l$, $l=1,...,2M+N$. Recall the
    definition of $\zeta(k)$, $k=1,2,...$ in section~{II}.
    Further,
     denote $P_i:=\mathrm{Prob}\left( \zeta(k)=i\right)>0$, $i=0,1,2,...,2M+N$.
     Then, it is easy to
     see that the P--AL--G algorithm can be formulated as in Algorithm~2.
\vspace{-1mm}
Finally, we summarize the overall primal-dual \textbf{AL--G}
 algorithm in Algorithm~3.
 \vspace{-5mm}
\begin{algorithm}
\caption{  \textbf{AL--G} algorithm at node $i$   }
\begin{algorithmic}[1]
    \STATE  Set $t=0$, $\lambda_{(i,j)} (t=0)=0$, ${\mu}_{(i,j)}(t=0)=0$, $j \in \Omega_i$ \\
    \REPEAT
    \STATE Run P--AL--G (cooperatively with the rest of the network) to get $x_i^{F}(t)$, $y_{ij}^{F}(t)$ and $y_{ji}^{L}(t)$, $j \in \Omega_i$
    \STATE
     Update the dual variables, $\lambda_{(i,j)} (t)$, ${\mu}_{(i,j)}(t)$, $j \in \Omega_i$, according to eqn.~(8).
    \STATE Set $t \leftarrow t+1$
    \UNTIL a stopping criterion is met.
\end{algorithmic}
\end{algorithm}
\vspace{-5mm}


\mypar{Remark} With \textbf{AL--G}, the updates of the primal variables, on a fast time scale $k$, are asynchronous
and use gossip communication, while the updates of the dual variables, on a slow time scale $t$,
 are \emph{synchronous} and require no communication. Physically, this can be realized as follows. Each (physical) node in the network has a timer, and the timers
 of different nodes are synchronized. At the beginning
 of each (slow time scale) $t$-slot,
 nodes start the gossip communication phase and cooperatively run the P--AL--G algorithm. After
 a certain predetermined time elapsed, nodes stop the communication phase and, during an
 idle communication interval, they update the dual variables. After the idle time elapses,
  the nodes restart the communication phase at the beginning of the new $t$-slot.

%

%
\mypar{Choice of $\rho$} It is known that, under Assumption~1, the method of multipliers~(6) converges
(i.e., any limit point of the sequence $x_i^\star(t)$, $t=0,1,...$,
   is a solution of~\eqref{eqn_original_primal_problem}) for any choice of the positive parameter $\rho$, \cite{RockafellarProof}, Theorem 2.1. It converges also if a sequence $\rho_{t+1} \geq  \rho_t$ is used, \cite{MultiplierSurvey}, Proposition~4. See \cite{BertsekasOptimization}, 4.2.2, for a discussion on the
choice of $\rho$. The method of multipliers still converges if we use different parameters $\rho = \rho_{\left( \lambda_{(i,j)},t\right)}$, $\rho = \rho_{\left( \mu_{(i,j)},t\right)}$, for each of the
  variables $\lambda_{(i,j)}$, $\mu_{(i,j)}$. This corresponds to
  replacing the quadratic terms $\rho\,\|x_i-y_{ij}\|^2$ and
  $\rho\,\|y_{ij}-y_{ji}\|^2$ in eqn. \eqref{eqn-lagrangian}
   with $\rho_{\left( \mu_{(i,j)},t\right)}\,\|x_i-y_{ij}\|^2$ and
  $\rho_{\left( \mu_{(i,j)},t\right)}\,\|y_{ij}-y_{ji}\|^2$, respectively.
  See reference \cite{AL_matrix_penalty} for details. (We still need $\rho_{\left( \lambda_{(i,j)},t\right)}
  \approx \rho_{\left( \lambda_{(j,i)},t\right)}$.) Equation \eqref{eqn_solving_wrt_x_i-simplified} becomes\footnote{Reference~\cite{AL_matrix_penalty} proves convergence
  of the method of multipliers with the positive definite matrix (possibly time-varying) penalty update, see eqn.~(1.5) in~\cite{AL_matrix_penalty}; the case of different (possibly time-varying) penalties assigned to different constraints is a special case
  of the matrix penalty, when the matrix is diagonal (possibly time-varying.)}
  \begin{equation}
 \begin{array}[+]{ll}
\mbox{minimize} & f_i(x_i) + \left(    \overline{\mu}_i - \sum_{j \in \Omega_i} \rho_{(\lambda_{(i,j)},t)}\, y_{ij}    \right)^{\top} x_i
 + \frac{1}{2} \left( \sum_{j \in \Omega_i} \rho_{(\mu_{(i,j)},t)} \right) \|x_i\|^2\\
\mbox{subject to} & x_i \in \mathcal{X}_i
\end{array}
\end{equation}
and equation~\eqref{eqn_solution_y_ij} becomes
\begin{equation}
y_{ij} = \frac{\rho_{(\lambda_{(i,j)},t)}}{\rho_{(\lambda_{(i,j)},t)}+\rho_{(\mu_{(i,j)},t)}}
\left( x_i+y_{ji}\right)+\frac{\mu_{(i,j)} - \mathrm{sign}(j-i)\lambda_{(i,j)}}{\rho_{(\lambda_{(i,j)},t)+\rho(\mu_{(i,j)},t)}}.
\end{equation}
  %

One possibility for adjusting the parameters $\rho_{(\mu_{(i,j)},t)}$ and $\rho_{(\lambda_{(i,j)},t)}$
in a distributed way is as follows. Each node $i$ adjusts (updates)
the parameters $\rho_{(\lambda_{(i,j),t})}$, $\rho_{(\mu_{(i,j),t})}$,
$j \in \Omega_i$. We focus on the parameter $\rho_{(\lambda_{(i,j)},t)}$; other parameters are updated similarly. Suppose that the current time is $t$. Node $i$ has stored in its memory the constraint violation at the previous time $t-1$ that equals $\epsilon_{(\lambda_{(i,j)},t-1)} = \|y_{ij}^F(t-1) -y_{ji}^L(t-1) \|$.
 Node $i$ calculates the constraint violation at the current time
 $\epsilon_{(\lambda_{(i,j)},t)} = \|y_{ij}^F(t) -y_{ji}^L(t) \|$.
  If $\epsilon_{(\lambda_{(i,j)},t)} / \epsilon_{(\lambda_{(i,j)},t-1)} \leq
  \kappa_{(\lambda_{(i,j)})}<1$,
  then the constraint violation is sufficiently decreased,
   and the parameter $\rho_{(\lambda_{(i,j)},t)}$
    remains unchanged, i.e., node $i$ sets $\rho_{(\lambda_{(i,j)},t)} = \rho_{(\lambda_{(i,j)},t-1)}$;
      otherwise, node $i$ increases the parameter, i.e., it sets
       $\rho_{(\lambda_{(i,j)},t)} =
      \sigma_{(\lambda_{(i,j)})} \rho_{(\lambda_{(i,j)},t-1)} $.
      The constants $\kappa_{(\lambda_{(i,j)})} \in (0,1)$
       and $\sigma_{(\lambda_{(i,j)})}>1$ are local
       to node $i$.
%
%
%
%
%
%
%

%
%
%
%
%
\vspace{-1.5mm}
\section{Convergence analysis of the \textbf{AL--G} algorithm}
\label{convergence_analysis}
This section provides the convergence of the \textbf{AL--G} algorithm. Convergence of the multiplier method for the dual variable updates (on slow time scale $\{t\}$) in eqn.~\eqref{eqn_dual_update_rule_lambda_mu} is available in the literature, e.g.,~\cite{BertsekasOptimization}. We remark that, in practice, P--AL--G runs for a finite time, producing an inexact solution of~\eqref{eqn_solve_for_x_i_y_ij}. This, however, does not violate the convergence of the overall primal-dual \textbf{AL--G} scheme, as corroborated by simulations in Section~\ref{section_simulations}. The P--AL--G algorithm for the primal variable update (on the fast time scale $\{k\}$) is novel, and its convergence requires a novel proof. We proceed with the convergence analysis of P--AL--G. First, we state an additional assumption on the function $f(\cdot)$, and we state Theorem~4 on the almost sure convergence of P--AL--G.

\vspace{-1.5mm}
\mypar{Assumptions and statement of the result}
Recall the equivalent definition of the P--AL--G and the simplified
 notation in~\ref{subsubsec-abst-model-of-p-sng}. The P--AL--G algorithm
   solves the following optimization problem:
   \begin{equation}
\begin{array}[+]{ll}
\mbox{minimize} & L(z) \\
\mbox{subject to} & z \in C
\end{array}.
\label{eqn_opt_problem_L(z)}
\end{equation}
We will impose an additional, mild assumption
 on the function $L(z)$, and, consequently, on the function $f(\cdot)$.
 First, we give the definition of a block-optimal point.
 \begin{definition}[Block-optimal point]
\label{definition-block-optimal}
A point $z^{\bullet} =  \left( z_1^{\bullet}, z_2^{\bullet}, ..., z_{N+2M}^{\bullet} \right)$ is block-optimal for the problem~\eqref{eqn_opt_problem_L(z)} if: $
z_i^{\bullet} \in \mathrm{arg \, min}_{w_i \in C_i} L \left( z_1^{\bullet},z_2^{\bullet},...,z_{i-1}^{\bullet},w_i,z_{i+1}^{\bullet},...,z_{N+2M}^{\bullet} \right), \,i=1,...,N+2M.
$
\end{definition}
 \begin{assumption}
 \label{ass_on_f}
 If a point $z^\bullet$ is a block-optimal point of~\eqref{eqn_opt_problem_L(z)},
  then it is also a solution of~\eqref{eqn_opt_problem_L(z)}.
 \end{assumption}
 \mypar{Remark} Assumption~\ref{ass_on_f} is mild: it is valid if, e.g., $f_i(x)=k_i \|x\|_1+W_i(x)$,
   $k_i \geq 0$, where $W_i: {\mathbb R}^m \rightarrow \mathbb R$ is a continuously differentiable, convex function, and
   $\|x\|_1=\sum_{i=1}^N |x_i|$ is the $l_1$ norm of~$x$, \cite{Tseng_new}.
 %
%
%

Define the set of optimal solutions $B=\{z^\star \in C:\,L(z^\star)=L^\star\}$, where $L^\star
=\inf_{z \in C} L(z)$.\footnote{Under Assumption~1, the set $B$ is nonempty and compact and $L^\star>-\infty$. This will be shown in Lemma~5. Clearly, $L^\star=L^\star \left(
\{\lambda_{(i,j)}\},\{\mu_{(i,j)}\}\right)$ and $B=B\left(
\{\lambda_{(i,j)}\},\{\mu_{(i,j)}\}\right)$ depend on the dual variables. For simplicity, we write only $L^\star$ and $B$.} Further, denote by $\mathrm{dist}(b,A)$ the Euclidean distance of
point $b \in {\mathbb R}^m$ to the set $A \subset {\mathbb R}^m$, i.e., $\mathrm{dist}(b,A) = \inf_{a \in A} \|a-b\|_2$, where $\|x\|_2$ is the Euclidean, $l_2$ norm. We now state the Theorem on almost sure (a.s.) convergence of the P--AL--G algorithm (Theorem~4,)
 after which we give some auxiliary Lemmas needed to prove Theorem~4.
\begin{theorem}
Let Assumptions~1 and 3 hold, and consider the optimization problem~\eqref{eqn_opt_problem_L(z)}
 (with fixed dual variables.) Consider the sequence $\left\{ z(k) \right\}_{k=0}^{\infty}$ generated by the algorithm P--AL--G.~Then:
\begin{enumerate}
\item
$
\lim_{k \rightarrow \infty} \mathrm{dist} \left(  z(k), B  \right) = 0,\,\,\mathrm{a.s.}
$
\item
$
\lim_{k \rightarrow \infty} L \left(  z(k)  \right) = L^\star,\,\,\mathrm{a.s.}
$
\end{enumerate}
\end{theorem}
%
%
%
%
%

\vspace{-1.5mm}
\mypar{Auxiliary Lemmas} Let $i_C: {\mathbb R}^m \rightarrow \mathbb R \cup \{+\infty\}$ be the indicator function of the set
$C$, i.e., $i_C(z)=0$ if $z \in C$ and $+\infty$ otherwise. It will be useful to define the
function $L+i_C:\,{\mathbb R}^{m(N+2M)} \rightarrow \mathbb R \cup \{+\infty\}$,
$(L+i_C)(z)=L(z)+i_C(z).$ Thus, the optimization problem~\eqref{eqn_opt_problem_L(z)} is equivalent to the unconstrained
minimization of $(L+i_C)(\cdot)$. The following Lemma establishes properties of the set of solutions $B$, the optimal value $L^\star$, and the function $(L+i_C)(\cdot)$.
\begin{lemma}
Let Assumption~1 hold. The functions $L(z)$ and $(L+i_C)(z)$ are coercive, $L^\star>-\infty$,
 and the set $B$ is nonempty and compact.
\end{lemma}
\begin{proof}
The function $L(z)$ (given in eqn.~\eqref{eqn-lagrangian}) is coercive. To see
this, consider an arbitrary sequence $\{z(j)\}_{j=1}^{\infty}$, where $\|z(j)\| \rightarrow \infty$ as $j \rightarrow \infty$.
 We must show that $L(z(j))\rightarrow \infty$. Consider two possible cases: 1) there is $i\in \{1,...,N\}$ such that $\|x_i(j)\|
 \rightarrow \infty$; and 2) there is no $i\in\{1,...,N\}$ such that $\|x_i(j)\| \rightarrow \infty$. For case 1), pick an $i$ such that $\|x_i(j)\|
 \rightarrow \infty$; then $f_i(x_i(j))\rightarrow \infty$, and hence, $L(z(j))\rightarrow \infty$. In case 2),
  there exists a pair $(i,l)$ such that $\|y_{il}\| \rightarrow \infty$; but then, as $x_i(j)$
   is bounded, we have that $\|x_i(j)-y_{il}(j)\|^2 \rightarrow \infty$, and hence, $L(z(j)) \rightarrow \infty.$
  The function
 $(L+i_C)(z)$ is coercive because $(L+i_C)(z)\geq L(z)$, $\forall z$, and $L(z)$
  is coercive.
  The function $(L+i_C)(z)$ is a closed\footnote{
  A function $q: {\mathbb R}^m \rightarrow \mathbb R \cup \{+\infty\}$ is closed if
  its epigraph $\mathrm{epi}(q) = \{(x,v):\,q(x)\leq v\}$ is a closed subset of ${\mathbb R}^{m+1}$.
  } (convex) function, because $L(z)$
   is clearly a closed function and $i_C(z)$ is a closed function because $C$ is a closed set; hence,
    $(L+i_C)(z)$ is closed function as a sum of two closed functions.Hence, $B$ is
      a closed set, as a sublevel set of the closed function $(L+i_C)(z)$. The set $B$ is bounded as a sublevel set of a
  coercive function $(L+i_C)(z)$. Hence, $B$ is closed and bounded, and thus, compact. We have that $L^\star>-\infty$ (and $B$ is non empty) as $L(z)$ is a continuous, convex, and coercive~on~${\mathbb R}^{m(N+2M)}$.
\end{proof}
Define $U_{\epsilon}(B) = \left \{ z: \mathrm{dist}(z,B)< \epsilon \right\}$, and let $V_{\epsilon}(B)$ be its complement, i.e.,
$V_{\epsilon}(B) = {\mathbb R}^m \backslash U_{\epsilon}(B)$. Further, denote by $S$ and $F$ the initial sublevel sets of $L$ and $L+i_C$, respectively, i.e.,
$S=\left\{ z: L(z) \leq L(z(0)) \right\}$, and $F = \left\{ z:\,\,(L+i_C)(z) \leq L(z(0)) \right\} = S \cap C$,
where $z(0) \in C$ is a feasible, deterministic, initial point. We remark that, given $z(0)$, any realization of the sequence $\left\{ z(k) \right\}_{k=0}^{\infty}$ stays inside the set $F$. This is true because $L(z(k))$
is a nonincreasing sequence by the definition of the algorithm P--AL--G and because any point $z(k)$ is feasible. Define also the set  $\Gamma(\epsilon)=F \cap V_{\epsilon}(B)$. We now remark that, by construction of the P--AL--G algorithm, the sequence of iterates $z(k)$ generated by P--AL--G is a Markov sequence.
%
%
%
%
%
We are interested in the expected decrease of the function $L(\cdot)$ in one algorithm step, given that the current point is equal to $z(k)=z$:
\begin{equation}
\psi(z) = \mathrm{E} \left[   L \left( z(k+1) \right)  | z(k)=z \right]-L(z).
\end{equation}
Denote by $L^i(z)$ the block-optimal value of the function $L(z)$ after minimization with respect to $z_i$:
\begin{equation}
\label{eqn_L^i(z)}
L^i(z) = \min_{w_i \in C_i} L \left( z_1,z_2,...,z_{i-1},w_i,z_{i+1},...,z_{N+2M} \right)
\end{equation}
We have, by the definition of P--AL--G, that (recall the definition of $P_i$ above Algorithm~2:)
\begin{eqnarray}
\psi \left( z \right)   &=& \sum_{i=1}^{N+2M} P_i\, \left( L^i(z) - L(z)  \right).
\label{eqn_expression_for_psi(z)}
\end{eqnarray}
%
%
%
Define $\phi(z) = - \psi(z)$. From eqn.~\eqref{eqn_expression_for_psi(z)}, it can be seen that $\phi(z) \geq 0$, for any $z \in C$. We will show that $\phi(z)$ is strictly positive on the set $\Gamma(\epsilon)$ for any positive $\epsilon$.
\begin{lemma}
\label{Lemma_inf}
\begin{equation}
\inf_{z \in \Gamma(\epsilon)} \phi(z) = a(\epsilon)>0
\label{eqn_inf_phi}
\end{equation}
\end{lemma}
We first show that $\Gamma(\epsilon)$ is compact and that $L^i$ is
continuous on $\Gamma(\epsilon)$ (latter proof is in the Appendix.)
\begin{lemma}[Compactness of $\Gamma(\epsilon)$]
The set $\Gamma(\epsilon)$ is compact, for all $\epsilon>0$.
\end{lemma}
\begin{proof}
We must show that $\Gamma(\epsilon)$ is closed and bounded. It is closed because it is the
intersection of the closed sets $F$ and $V_{\epsilon}(B)$. It is bounded because $\Gamma(\epsilon) \subset F$, and $F$ is bounded. The set $F$ is bounded as a sublevel set of the coercive function $L+i_C$. 
The set $F$ is closed as a sublevel set of the closed function $L+i_C$.
\end{proof}
\begin{lemma}[Continuity of $L^i$]
The function $L^i: \,\Gamma(\epsilon) \rightarrow {\mathbb R}$ is continuous, $i=1,...,N+2M$.
\end{lemma}
%
%
%
\begin{proof}[Proof of Lemma~\ref{Lemma_inf}]
First, we show that $\phi(z)>0$, for all $z \in \Gamma (\epsilon)$. Suppose not. Then, we have: $L^i(z) = L(z)$, for all $i$. This means that the point $z \in \Gamma(\epsilon)$ is block-optimal; Then, by
Assumption~3, the point $z$ is an optimal solution of~\eqref{eqn_opt_problem_L(z)}. This is a contradiction and $\phi(z)>0$, for all $z \in \Gamma(\epsilon)$.
Consider the infimum in eqn.~\eqref{eqn_inf_phi}. The infimum is over the compact set and the function $\phi(\cdot)$ is continuous (as a scaled sum of continuous functions $L^i(\cdot)$); thus, by the Weierstrass theorem, the infimum is attained for some $z^{\bullet} \in \Gamma(\epsilon)$ and $\phi(z^{\bullet}) = a(\epsilon) >0$.
\end{proof}
%
%
%
%
%
%
%
%
\vspace{-1mm}
\mypar{Proof of Theorem~4--1} Recall the expected decrease of the function $L(\cdot)$, $\psi(z)$. We have:
\begin{eqnarray}
\mathrm{E} \left[ \psi\left( z(k) \right) \right]=\mathrm{E}  \left[  \mathrm{E} \left[ L \left( z(k+1) \right) | z(k) \right] - L \left( z(k) \right) \right]
= \mathrm{E} \left[  L \left( z(k+1) \right)   \right] - \mathrm{E} \left[  L \left( z(k) \right)   \right].
\end{eqnarray}
On the other hand, we have that $\mathrm{E}\left[ \psi(z(k))\right]$ equals:
\begin{eqnarray}
 \mathrm{E} \left[  \psi \left( z(k) \right) | z(k) \in {\Gamma(\epsilon) }\right] \,
\mathrm{Prob} \left( z(k) \in {\Gamma(\epsilon )}\right) + \mathrm{E} \left[  \psi \left( z(k) \right) | z(k) \notin {\Gamma(\epsilon) }\right] \,\mathrm{Prob} \left( z(k) \notin {\Gamma(\epsilon )}\right).
\end{eqnarray}
Denote by $p_k = \mathrm{Prob} \left( z(k) \in \Gamma(\epsilon) \right)$.
Since $\psi\left( z(k) \right) \leq -a (\epsilon)$, for $z(k) \in {\Gamma(\epsilon)}$, and $\psi\left( z(k) \right) \leq 0$, for any $z(k) $, we have that:
%
$\mathrm{E} \left[  \psi \left( z(k) \right) \right] = \mathrm{E} \left[  L(z(k+1))   \right]  - \mathrm{E} \left[ L(z(k)) \right]\leq -a(\epsilon)\,p_k;$
%
summing up latter inequality for $j=0$ up to $j=k-1$, we get:
\begin{equation}
\mathrm{E} \left[  L(z(k)) \right] -   L(z(0))  \leq  -a(\epsilon)  \sum_{j=0}^{k-1}p_k, \, \forall j \geq 0.
\end{equation}
The last inequality implies that: $
\sum_{k=0}^{\infty} p_k \leq \frac{1}{a(\epsilon)} \, \left(  L \left( z(0) \right) -L^\star \right) < \infty.
$
Thus, by the first Borel-Cantelli Lemma, $\mathrm{Prob} \left( z(k) \in {\Gamma(\epsilon) ,\,\,\mathrm{infinitely\,\,often}}\right) = 0$, $\forall \epsilon>0$. Thus, $\mathrm{Prob} \left( \mathcal{A}_{\epsilon}\right)=1$, $\forall \epsilon>0$,
 where the event $\mathcal{A}_{\epsilon}$ is: $\mathcal{A}_{\epsilon}:=\left\{ \mathrm{the \,\,tail\,\,of\,\,the\,\,sequence\,\,}z(k) \mathrm{\,\,belongs\,\,to\,\,}
 U_{\epsilon}(B) \right\}$. Consider the event $\mathcal{A}:=\cap_{s=1}^\infty \mathcal{A}_{\epsilon_s}$,
  where $\epsilon_s$ is a decreasing sequence, converging to $0$. Then, $\mathrm{Prob}
  \left( \mathcal{A}\right)=\mathrm{Prob}\left( \cap_{s=1}^\infty \mathcal{A}_{\epsilon_s}\right)=
  \lim_{s \rightarrow \infty} \mathrm{Prob}
  \left( \mathcal{A}_{\epsilon_s}\right)=\lim_{s \rightarrow \infty} 1=1$. Now, the event $\mathcal{B}:=\left\{
 \lim_{k \rightarrow \infty}\mathrm{dist}(z(k),B)=0\right\}$ is equal to
 $\mathcal{A}$, and thus $\mathrm{Prob}\left( \mathcal{B}\right)=1$.
\vspace{-1.2mm}

\mypar{Expected number of iterations for convergence: Proof of Theorem~4--2} Consider now the sets $\mathcal{U}_{\epsilon}(B) = \left \{ z: L(z) \leq  \epsilon + L^\star \right \}$ and $\mathcal{V}_{\epsilon}(B) = {\mathbb R}^m \backslash \mathcal{U}_{\epsilon}(B)$ and define the sets $\mathcal{F}$ and $\mathcal{G}(\epsilon)$ as $\mathcal{F} = C \cap S$ and
 $\mathcal{G}(\epsilon) = \mathcal{F} \cap \mathcal{V}_{\epsilon}(B)$. Similarly as in Lemmas~8, we can
  obtain that
 \begin{equation}
 \inf_{z \in \mathcal{G}(\epsilon)} \phi(z) =  {b}(\epsilon) >0.
 \end{equation}
We remark that, once $z(k)$ enters the set $\mathcal{U}_{\epsilon}(B)$ at $k=K_{\epsilon}$, it never leaves this set, i.e., $z(k) \in \mathcal{U}_{\epsilon}(B)$, for all $k \geq K_{\epsilon}$. Of course, the
 integer $K_{\epsilon}$ is random. In the next Theorem, we provide an upper bound on the expected
 value of $K_{\epsilon}$ (the time slot when $z(k)$ enters the set $\mathcal{U}_{\epsilon}(B)$,) thus giving a stopping criterion (in certain sense) of the algorithm P--AL--G.
 \vspace{-3mm}
\begin{theorem}[  Expected number of iterations for convergence   ]
Consider the sequence $\left\{ z(k) \right\}_{k=0}^{\infty}$ generated by the algorithm P--AL--G. Then, we have:
\begin{equation}
\label{eqn-exp-num-iter}
\mathrm{E} \left[ K_{\epsilon} \right] \leq \frac  { L  \left(  z(0) \right) - L^\star} {b({\epsilon})}.
\end{equation}
\end{theorem}
\begin{proof}
Let us define an auxiliary sequence $\widetilde{z}(k)$ as
$\widetilde{z}(k) =  z(k) $, if $z(k) \in \mathcal{G}(\epsilon) $, and
 $ z(k) = z^{\star}$, if $z(k) \in \mathcal{U}_{\epsilon}(B) $. Here $z^{\star}$ is a point in $B$.
 That is, $\widetilde{z}(k)$ is identical to $z(k)$ all the time while $z(k)$ is outside the set $\mathcal{U}_{\epsilon}(B)$ and $\widetilde{z}(k)$ becomes $z^{\star}$
 and remains equal to $z^{\star}$ once $z(k)$ enters $\mathcal{U}_{\epsilon}(B)$. (Remark that $z(k)$ never leaves the set $\mathcal{U}_{\epsilon}(B)$ once it enters it by construction of Algorithm P--AL--G.)

 Now, we have that:
 \begin{equation}
  \psi \left( \widetilde{z}(k) \right) = \left\{ \begin{array}{rl}
 \psi (z(k)) \leq -b(\epsilon) &\mbox{ if $z(k) \in \mathcal{G}(\epsilon) $} \\
 0 &\mbox{ if $z(k) \in \mathcal{U}_{\epsilon}(B)$}
       \end{array} . \right.
 \end{equation}
Taking the expectation of $\psi \left(z(k) \right)$, $k=0,...,t-1$ and summing up these expectations, and letting $t \rightarrow \infty$, we get:
\[
\mathrm{E} \left[  L \left( \widetilde{z}(\infty) \right) \right] - L(z(0))  = \sum_{k=0}^{\infty} \mathrm{E} \left[ \psi \left( \widetilde{z}(k+1) \right) \right]  - \mathrm{E} \left[ \psi \left( \widetilde{z}(k) \right) \right]\\
= \mathrm{E} \sum_{k=0}^{\infty} \psi \left( \widetilde{z}(k) \right) \leq - \mathrm{E} \left[ K_{\epsilon} \right]\,b(\epsilon)
\]
Thus, the claim in equation~\eqref{eqn-exp-num-iter} follows.
\end{proof}
%
We now prove Theorem~4--2. By Theorem 10, the expected value of $K_{\epsilon}$ is finite, and thus $K_{\epsilon}$ is finite a.s. This means that for all $\epsilon>0$, there exists random number $K_{\epsilon}$ (a.s. finite), such that $\widetilde{z}(k) = z^{\star}$, for all $k \geq K_{\epsilon}$, i.e., such that
$z(k) \in \mathcal{U}_{\epsilon}(B)$ for all $k \geq K_{\epsilon}$. The last statement is equivalent to Theorem~4--2.
%
%
%
%
\vspace{-3mm}
\section{Variants to \textbf{AL--G}: \textbf{AL--MG} (augmented Lagrangian multi neighbor gossiping) and \textbf{AL--BG} (augmented Lagrangian broadcast gossiping) algorithms}
\label{subsection_MNG}
This section introduces two variants to the \textbf{AL--G} algorithm, the \textbf{AL--MG} (augmented Lagrangian multi neighbor gossiping)
 and the \textbf{AL--BG} (augmented Lagrangian broadcast gossiping). Relying on the previous description and analysis of the \textbf{AL--G} algorithm, this section explains specificities of the \textbf{AL--MG} and \textbf{AL--BG} algorithms.
  Subsection~{V-A} details the \textbf{AL--MG}, and subsection~{V-B} details
  the \textbf{AL--BG} algorithm.
   Proofs of the convergence for P--AL--MG and P--AL--BG are in the Appendix.
\vspace{-3mm}
\subsection{\textbf{AL--MG} algorithm}
\label{subsect_MNG_alg}
The \textbf{AL--MG} algorithm is a variation of the \textbf{AL--G} algorithm. The algorithms \textbf{AL--G} and \textbf{AL--MG} are based on the same reformulation of~\eqref{eqn_original_primal_problem} (eqn.\eqref{eqn_primal_reformulated_with_y_s}), and they have the same dual variable update (that is, D--AL--G and D--AL--MG are the same.) We proceed
by detailing the difference between P--AL--MG and P--AL--G to solve~\eqref{eqn_solve_for_x_i_y_ij} (with fixed dual variables.) With
  the algorithm~P--AL--MG, each node has two independent Poisson clocks, $T_i^x$ and $T_i^y$.
  Update followed by a tick of $T_i^x$ is the same as with P--AL--G (see Algorithm~1, step~4.) If $T_i^y$
 ticks, then node $i$ transmits \emph{simultaneously} the variables $y_{ij}$, $j \in \Omega_i$, to all its neighbors ($y_{i,j_1}$ is transmitted to node $j_1$, $y_{i,j_2}$ is transmitted to node $j_2$, etc.) Due to link failures, the neighborhood nodes may or may not receive the transmitted information. Successfull transmissions are followed by updates of $y_{ji}$'s, according to
  eqn.~\eqref{eqn_solution_y_ij}.
%
Define also the virtual clock $T$
 that ticks whenever one of the clocks $T_i^x$, $T_i^y$, ticks. Accordingly, we define the $k$-time slots
  as $[\tau_{k-1},\tau_k)$, $k=1,2...$, $\tau_0=0$, and $\tau_k$ is the time of the $k$-th tick of $T$.
Overall \textbf{AL--MG} algorithm is the same as \textbf{AL--G} (see Algorithm~3,) except that, instead
of P--AL--G, nodes run P--AL--MG algorithms at each $t$. We prove
convergence of the P--AL--MG in the Appendix; for convergence of the overall \textbf{AL--MG} algorithm, see discussion at the beginning of section~V.
\vspace{-3mm}
\subsection{\textbf{AL--BG} algorithm: An algorithm for static networks}
\label{subsection_BG} We now present a simplified algorithm for the networks with reliable transmissions. This algorithm is based on the
reformulation of~\eqref{eqn_original_primal_problem} that eliminates the variables $y_{ij}$'s. That is, we start with the following
equivalent formulation of~\eqref{eqn_original_primal_problem}:
\begin{equation}
\begin{array}[+]{ll}
\mbox{minimize} & \sum_{i=1}^N f_i(x_i) \\
\mbox{subject to} & x_i \in \mathcal{X}_i,\,\,i=1,...,N,
\\ & x_i = x_j,\,\, \{i,j\} \in E
\end{array}
\label{eqn_reformulated_primal_no_y_s}
\end{equation}
We remark that~\eqref{eqn_reformulated_primal_no_y_s} is equivalent to~\eqref{eqn_original_primal_problem} because the
 supergraph is connected. After dualizing the constraints $x_i = x_j$, $(i,j) \in E$, the AL dual function
 $L_a(\cdot)$ and the Lagrangian $L(\cdot) $ become:
\begin{eqnarray*}
\begin{array}[+]{ll}
L_a \left(  \{ \lambda_{\{i,j\} } \} \right) \,= \,
\mbox{\,\,\,\,\,\,\,\,min} & L \left(  \{x_i\},\, \{    \lambda_{\{i,j\}}  \}  \right) \\
\,\,\,\,\,\,\,\,\,\,\,\,\,\,\,\,\,\,\,\,\,\,\,\,\,\,\,\,\,\,\,\,\,
\,\,\,\,\,\,\,\,\,\,\,\,\,\,\,\,\,
\mbox{subject to} & x_i \in \mathcal{X}_i,\, i=1,...,N
\end{array}
\end{eqnarray*}
%
\begin{equation}
\label{eqn-l-bg}
L \left(  \{x_i\},\, \{\lambda_{\{i,j\}} \}  \right) = \sum_{i=1}^N f_i(x_i) + \sum_{\{i,j\} \in E,\,i<j} \lambda_{\{i,j\}}^{\top}\, \left( x_i - x_j \right)
 + \frac{1}{2} \rho \sum_{\{i,j\} \in E, \,i<j}\,\| x_i - x_j \|^2.
\end{equation}
In the sums
 $ \sum_{ \{i,j\} \in E} \lambda_{\{i,j\}}^{\top}\, \left( x_i - x_j \right)$ and
 $\sum_{\{i,j\} \in E}\,\| x_i - x_j \|^2$, the terms
 $\lambda_{ \{i,j\} }^{\top}\, \left( x_i - x_j \right)$ and $\| x_i - x_j \|^2$ are included once. (The summation is over the undirected edges $\{i,j\}$.) Also, terms $\lambda_{\{i,j\}}^{\top}\, \left( x_i - x_j \right)$ in the sum
$\sum_{ \{i,j\} \in E}
\lambda_{ \{i,j\} }^{\top}
\, \left( x_i - x_j \right)$ are arranged such that $i<j$, for all $\{i,j\} \in E$. The resulting dual optimization problem is the unconstrained maximization of $L_a(  \lambda_{ \{ i,j \} } )$.
%
%

\mypar{Solving the dual: D--AL--BG algorithm} We solve the dual~\eqref{eqn-l-bg} by the method of multipliers, which can be shown to have the following form:
\begin{equation}
\lambda_{\{i,j\}}(t+1) = \lambda_{\{i,j\}}(t) + \rho \, \mathrm{sign}(j-i)\, \left( x_i^{\star}(t) - x_j^{\star}(t) \right)
\label{eqn_BG_dual_var_update}
\end{equation}
\vspace{-2mm}
\begin{eqnarray}
\label{eqn_solve_for_x_problem}
\begin{array}[+]{ll}
x^{\star}(t) = \left(x_1^{\star}(t), x_2^{*}(t),...,x_N^{*}(t) \right)\, \in \,
\mbox{arg min} & L \left(  \{x_i\},\, \{\lambda_{\{i,j\}} (t) \}  \right) \\
\,\,\,\,\,\,\,\,\,\,\,\,\,\,\,\,\,\,\,\,\,\,\,\,\,\,\,\,\,\,\,\,\,\,\,\,\,\,\,
\,\,\,\,\,\,\,\,\,\,\,\,\,\,\,\,\,\,\,\,\,\,\,\,\,\,\,\,\,\,\,\,\,\,\,\,\,\,\,\,\,\,\,\,\,\,\,\,\,
\mbox{subject to} & x_i \in \mathcal{X}_i,\, i=1,...,N
\end{array}.
\end{eqnarray}
We will explain in the next paragraph how the P--AL--BG algorithm solves~\eqref{eqn_solve_for_x_problem}
in a distributed, iterative way. With \textbf{AL--BG}, each node needs to maintain only one $m$-dimensional
 dual variable:
 $\overline{\lambda}_i := \sum_{ j \in \Omega_i} \mathrm{sign}(j-i)\,\lambda_{\{i,j\}}$.
  Also, define $\overline{x}_i:=\sum_{j \in \Omega_i} x_j$. The P--AL--G algorithm terminates after a finite number of inner iterations $k$, producing an inexact
  solution. Denote by $x_i^F$ (resp. $x_j^F$) the inexact solution of $x_i$ (resp. $x_j$, $j \in \Omega_i$), available at node $i$, after termination of P--AL--BG.
   We will see that $x_i^F = x_i^L$, $\forall i$; accordingly, after termination of P--AL--BG, node $i$ has available $\overline{x}_i^F:=\sum_{j \in \Omega_i} x_j^F$. Summing up equations~\eqref{eqn_BG_dual_var_update} for $\lambda_{\{i,j\}},\,j \in \Omega_i$, and taking into account the finite time termination of the P--AL--BG, we arrive at the following
     dual variable update at node $i$:
\begin{equation}
\overline{\lambda}_i(t+1)=\overline{\lambda}_i(t) + \rho \, \left( d_i\,x_i^{F}(t) - \overline{x}^{F}_i (t) \right),\,i=1,...,N.
\label{eqn_update_lambda_bar}
\end{equation}

\mypar{Solving for~\eqref{eqn_solve_for_x_problem}: P--AL--BG algorithm}
\label{subsubsect_BG_solving_primal}
We solve the problem~\eqref{eqn_solve_for_x_problem} by a randomized, block-coordinate P--AL--BG algorithm. After straightforward calculations, it can be shown that minimization of the function in~\eqref{eqn-l-bg} with respect to $x_i$ (while other coordinates are fixed) is equivalent to the following minimization:
 \begin{equation}
\begin{array}[+]{ll}
\mbox{minimize} & f_i(x_i) + \left(  \overline{\lambda}_i -  \rho \, \overline{x}_i \right)^{\top} x_i + \frac{1}{2} \,\rho \,d_i \|x_i\|^2 \\
\mbox{subject to} & x_i \in \mathcal{X}_i
\end{array}
\label{eqn_minimize_with_resp_x_i-simplif}
 \end{equation}
Similarly as with \textbf{AL--G}, we assume that the clock ticks at all nodes are governed by independent Poisson process $T_i$'s.
P--AL--BG is as follows. Whenever
clock $T_i$ ticks, node $i$ updates $x_i$ via eqn.~\eqref{eqn_minimize_with_resp_x_i-simplif} and broadcasts the updated $x_i$ to all
the neighbors in the network. Discrete random iterations $\{k\}$ of the P--AL--BG algorithm are defined as ticks of the virtual clock $T$ that ticks whenever one of $T_i$ ticks. The P--AL--BG algorithm produces $x_i^F$ and $\overline{x}_i^F$ at node $i$.
%
%
%
%
Overall primal-dual \textbf{AL--BG} algorithm is similar to the \textbf{AL--G} algorithm (see Algorithm~3), except that,
at each $t$, nodes cooperatively run the P--AL--BG algorithm, instead of P--AL--G algorithm.
We prove convergence of P--AL--BG in the Appendix; for convergence of the overall primal-dual
\textbf{AL--BG} scheme, see discussion at the beginning of Section~V.
\vspace{-3mm}
\section{Simulation examples}
\label{section_simulations}
In this section, we consider two simulation examples,
 namely, $l_1$--regularized logistic regression for classification (subsection~VI-A,) and cooperative spectrum sensing for cognitive radio networks (subsection~VI-B.) Both examples corroborate the convergence of our algorithms \textbf{AL--G}, \textbf{AL--MG} on random networks, and \textbf{AL--BG} on static networks, and demonstrate tradeoffs that our
  algorithms show with respect to the existing literature. We compare the convergence speed of our and existing algorithms
 with respect to: 1) communication cost; and 2) computational cost, while the communication cost
 is dominant in networked systems supported by wireless communication. \textbf{AL--BG}
     outperforms existing algorithms (in~\cite{asu-new-jbg,asu-random,johansson,bazerque_lasso}\footnote{Reference~\cite{bazerque_lasso} focusses specifically on the Lasso problem; we compare with~\cite{bazerque_lasso} in subsection {VI}-B.}) on static networks
     in terms of communication cost, on both examples; at the same time, it has a larger computational cost.
     For the $l_1$-regularized logistic regression example and random networks,
     \textbf{AL--G} and \textbf{AL--MG} outperform existing algorithms (\cite{asu-new-jbg,asu-random}\footnote{Only references~\cite{asu-new-jbg,asu-random} consider random
      networks.}) in terms of communication cost, while having larger
      computational cost. For the cooperative spectrum sensing example and random networks, \textbf{AL--G} and \textbf{AL--MG} converge slower than existing algorithms~\cite{asu-new-jbg,asu-random}.
\vspace{-3mm}
\subsection{$l_1$--regularized logistic regression for classification}
We consider distributed learning of a linear discriminant function. In particular,
 we consider the $l_1$--regularized logistic regression optimization problem (eqn.~(45) in~\cite{BoydADMoM}; see Subsections 7.1 and 10.2). We add private constraints and adapt the notation from~\cite{BoydADMoM} to fit our exposition.\footnote{Note that~\cite{BoydADMoM} studies only the parallel network architecture, with a fusion center, and it does not propose an algorithm to solve the $l_1$--regularized logistic regression problem on generic networks, the case that we address here.} The problem setup is as follows. Each node $i$, $i=1,...,N$,
 has $N_d$ data samples, $\left\{a_{ij},b_{ij}\right\}_{j=1}^{N_d}$,
 where $a_{ij} \in {\mathbb R}^m$ is a feature vector (data vector,) and $b_{ij}\in \{-1,+1\}$ is the
 class label of the feature vector $a_{ij}$. That is, when $b_{ij}=1$ (respectively, $-1$,) then the feature vector $a_{ij}$ belongs to the class ``1" (respectively, ``$-1$".) The goal is to learn the weight vector $w \in {\mathbb R}^m$, and
 the offset $v \in \mathbb R$, based on the available samples at all nodes, $\left\{a_{ij},b_{ij}\right\}_{j=1}^{N_d}$, $i=1,...,N$, so that $w$ is sparse, and the equality: $
  \mathrm{sign} \left( a_{ij}^\top w + v\right) = b_{ij},\,\,i=1,...,N,\,j=1,...,{N_d},
  $
  holds for the maximal possible number of data samples $\left\{a_{ij},b_{ij}\right\}_{j=1}^{N_d}$, $i=1,...,N$. One approach to choose $w$ and $v$ is via $l_1$--regularized logistic regression; that is, choose
$w^\star$ and $v^\star$ that solve the following optimization problem,~\cite{BoydADMoM}:
\begin{equation}
\label{eqn-logistic-reg-problem}
\begin{array}[+]{ll}
\mbox{minimize} & \sum_{i=1}^N \sum_{j=1}^{N_d} \log \left(  1+ \mathrm{exp} \left( -b_{ij}(a_{ij}^\top w + v)\right)\right)
+ \lambda \|w\|_1\\
\mbox{subject to} & w^\top w \leq k_i,\,\,i=1,...,N\\
                  & |v| \leq k_i^\prime,\,i=1,...,N
\end{array}.
\end{equation}
The parameter $\lambda>0$ enforces the sparsity in $w$, \cite{BishopBook}.
The private constraints on $w$ and $v$ at node $i$ ($k_i$'s and $k_i^\prime$'s are positive) represent the prior
knowledge available at node $i$ (see \cite{BoydBook}, Chapter 7.) Problem~\eqref{eqn-logistic-reg-problem} clearly
 fits our generic framework in~(1) and has a vector optimization variable, a non smooth objective function, and quadratic private constraints. Alternatives to~\eqref{eqn-logistic-reg-problem} to learn $w$ and $v$ include support
 vector machines and boosting,~\cite{BishopBook, BoydADMoM}. 

 \mypar{Simulation setup} We consider a supergraph with $N=20$ nodes and $|E|=37$ undirected edges ($74$ arcs). Nodes are uniformly distributed on a unit square
and pairs of nodes with distance smaller than a radius $r$
are connected by an edge. For networks with link failures, the link failures of different arcs at the same time slot are
independent and the failures of the same arc at different time slots
are independent also. Link failure probabilities $\pi_{ij}$
are generated as follows:
$
\pi_{ij} = k \,  \frac{ \delta_{ij}^2 } { r^2  }, \,\,
\delta_{ij}<r,
$ {where $k=0.5.$ Each node has $N_d=5$ data samples.
 Each feature vector $a_{ij} \in {\mathbb R}^m$, $m=20$, and the ``true" vector $w_{\mathrm{true}}$ have approximately $60\%$
  zero entries. Nonzero entries of $a_{ij}$ and $w_{\mathrm{true}}$, and the offset $v_{\mathrm{true}}$
  are generated independently, from the standard normal distribution. Class labels $b_{ij}$
    are generated by: $b_{ij} = \mathrm{sign}  \left( a_{ij}^\top w_{\mathrm{true}} + v_{\mathrm{true}} + \epsilon_{ij}\right)$, where $\epsilon_{ij}$ comes from the normal distribution with zero mean and variance $0.1$.
    The penalty parameter $\lambda$ is set to be $0.5 \cdot \lambda_{\mathrm{max}}$, where
    $\lambda_{\mathrm{max}} $ is the maximal value of $\lambda$ above which the solution
    to \eqref{eqn-logistic-reg-problem} is $w^\star=0$ (see (\cite{BoydADMoM}, subsection 10.2) how to find $\lambda_{\mathrm{max}}$.)
     We set $k_i$ and $k_i^\prime$ as follows. We solve the unconstrained version of \eqref{eqn-logistic-reg-problem} via
      the centralized subgradient algorithm; we denote the corresponding solution by $w^\bullet$ and $v^\bullet$.
       We set $k_i = (1+r_i) \cdot \|w^\bullet\|^2$, $k_i^\prime = (1+r_i^\prime) \cdot |v^\bullet|$,
        where $r_i$ and $r_i^\prime$ are drawn from the uniform distribution on $[0,1]$.
        Thus, the solution to problem \eqref{eqn-logistic-reg-problem}
         is in the interior of the constraint set. (Similar numerical results to the ones presented are obtained when
          the solution is at the boundary.) To
    update $x_i$ with P--AL--G and P--AL--MG (eqn.~\eqref{eqn_solving_wrt_x_i-simplified}), we solve~\eqref{eqn_solving_wrt_x_i-simplified}
     via the projected subgradient algorithm.

\mypar{Algorithms that we compare with} In the first set of experiments, we consider
\textbf{AL--BG} for (static) networks; in the second set of experiments, we test \textbf{AL--G} and \textbf{AL--MG} on networks with link failures. We compare our algorithms with the ones proposed in~\cite{asu-new-jbg,nedic_T-AC,nedic_novo,asu-random}\footnote{We simulate the algorithms in~\cite{asu-new-jbg,nedic_T-AC,nedic_novo,asu-random} with symmetric link failures.} and in~\cite{johansson}.
References~\cite{asu-new-jbg, nedic_T-AC,nedic_novo,asu-random} propose a primal projected subgradient algorithm, here refer to as PS (Primal
Subgradient.) PS, as an intermediate step,
computes weighted average of the optimal point estimates across node $i$'s neighborhood. Averaging weights have not been recommended
in~\cite{asu-new-jbg,nedic_T-AC,nedic_novo,asu-random}; we use the standard time-varying Metropolis weights, see~\cite{BoydFusion}, eqn.~\eqref{eqn_solving_wrt_x_i-simplified}. Reference~\cite{johansson} proposes an incremental primal subgradient
algorithm, here referred to as MCS (Markov chain subgradient.) With MCS, the order of incremental
 updates is guided by a Markov chain,~\cite{johansson}.\footnote{Convergence for MCS has been
 proved only with the projection onto a public constraint set, but we simulate it here
 with the straightforward generalization of the projection onto private constraint sets; MCS showed convergence for our example in
 the private constraints case also.} We simulate MCS and PS with fixed subgradient step size rather than the diminishing step size, as the former yields faster convergence.

We compare the algorithms based on two criteria. The first is the amount of inter-neighbor communication that the algorithms require to meet a certain accuracy. We count the total number of radio transmissions (counting
 both successful and unsuccessful transmissions.) The second is the total number of floating point operations (at all nodes.)
  In networked systems supported by wireless communication (e.g., WSNs,) the dominant cost (e.g., power consumption) is induced by communication. Total number of floating point operations depends on the algorithm implementation, but the results to be presented give a good estimate of the algorithms' computational cost. It may be possible to reduce the computational cost of
  \textbf{AL--G}, \textbf{AL--MG}, and \textbf{AL--BG} by a more computationally efficient solutions
  to problems \eqref{eqn_solving_wrt_x_i-simplified} and~(32) than (here adopted) projected subgradient method.

Denote by $f^\star$ the optimal value of \eqref{eqn-logistic-reg-problem}. We compare the algorithms in terms of the following metric:
\[
\textbf{err}_f = \frac{1}{N} \sum_{i=1}^N \left( f(x_i)-f^\star \right),
\]
where $x_i$ is the estimate of the optimal solution available at
node $i$ at a certain time.

With our \textbf{AL--G, AL-MG}, and \textbf{AL--BG} algorithms, the simulations to be presented use
 an increasing sequence of AL penalty parameters (see the end of Section~{IV},)
 which, after some experimentation, we set to the following values: $\rho_{t}=t^{A_{\rho}}+B_{\rho}$, $t=0,1,...$,
 with $A_{\rho}=1.3$, and $B_{\rho}=1$. We also
 implemented the algorithms with different and increasing $\rho$'s
  assigned to each dual variable, with the scheme for adjusting $\rho$'s explained at the end of Section~{IV}, with $\kappa_{\lambda_{(i,j)}}=\kappa_{\mu_{(i,j)}}=0.3$, and $\sigma_{\lambda_{(i,j)}}=\sigma_{\mu_{(i,j)}}
    = 1.2$. The latter choice also showed convergence of \textbf{AL--G, AL-MG}, and \textbf{AL--BG},
    but the former yielded faster convergence.
    Our simulation experience shows that
     the convergence speed of \textbf{AL--G, AL-MG}, and \textbf{AL--BG}
      depend on the choice of $\rho_t$, but the optimal tuning of $\rho_t$ is left for future studies. With PS and MCS, and a
 fixed step size, the estimates $f(x_i)$ converge only to a neighborhood of $f^\star$. There is a tradeoff between the limiting error $\textbf{err}_f(\infty)$ and the rate of convergence with respect to the stepsize $\alpha$: larger $\alpha$ leads to faster convergence and larger $\textbf{err}_f(\infty)$. We notice by simulation that \textbf{AL--G, AL--MG}, and \textbf{AL--BG}
  converge to a plateau neighborhood of $f^{\star}$; after that,
  they improve slowly; call the error that corresponds to this plateau $\textbf{err}_f(ss)$. To make
  the comparison fair or in favor of PS and MCS, we set $\alpha$ for the PS and MCS algorithms such that
  the $\textbf{err}_f(\infty)$ for PS and MCS is equal (or greater) than the $\textbf{err}(ss)$ attained by \textbf{AL--G, AL--MG}, and \textbf{AL--BG}.

\mypar{Results: Static network} Figure 2 (top left) plots
$\textbf{err}_f$ versus the number of ($m=20$-dimensional vector) transmissions (cumulatively at all nodes.)
We can see that \textbf{AL--BG} outperforms PS and MCS by one to two
 orders of magnitude. \textbf{AL--BG} needs
 about $0.3 \cdot 10^5$ transmissions to reduce $\textbf{err}_f$ below $0.001$, while MCS and PS need, respectively, about $4 \cdot 10^5$ and $18 \cdot 10^5$ transmissions
 for the same precision. With respect to the number of floating point
  operations (Figure 2, top right,) \textbf{AL--BG} needs more
  operations than MCS and PS; $45 \cdot 10^8$ for \textbf{AL--BG}
  versus $13 \cdot 10^8$ for PS, and $2 \cdot 10^8$ for MCS. Thus,
  with respect to MCS, \textbf{AL--BG} reduces communication at a cost of additional
  computation. Note that with \textbf{AL--BG}, MCS, and PS, due to private constraints, node $i$'s
  estimate $x_i$ may not be feasible at certain time slots;
 in this numerical example, \textbf{AL--BG}, MCS, and PS all produced feasible solutions
 at any time slot, at all nodes. A drawback of MCS in certain applications,
  with respect to PS and \textbf{AL--BG}, can be the \emph{delay time} that MCS needs for
  the ``token" to be passed from node to node as MCS evolves, see~\cite{johansson}.

\mypar{Results: Random network} Figure 2 (bottom left) plots $\textbf{err}_f$ versus
the total number of transmissions. \textbf{AL--MG} and \textbf{AL--G} outperform PS. To decrease $\textbf{err}_f$ below $5 \cdot 10^{-4}$,
\textbf{AL--MG} and \textbf{AL--G} require about $1.2 \cdot 10^6$
 transmissions, and \textbf{AL--G} $1.5 \cdot 10^6$ transmissions; PS requires about
  $3.7 \cdot 10^6$ transmissions to achieve the same precision.
  Figure 2 (bottom right) plots $\textbf{err}_f$
   plots versus the total number of floating point operations.
    PS requires less computation
    than \textbf{AL--G} and \textbf{AL--MG}.
    To decrease $\textbf{err}_f$ below $5 \cdot 10^{-4}$,
\textbf{AL--MG} and \textbf{AL--G}
  require about $69 \cdot 10^{9}$ transmissions; PS requires about
  $2.8 \cdot 10^9$ transmissions for same precision. With each of the algorithms \textbf{AL--G}, \textbf{AL--MG}, and PS, each node $i$'s estimate $x_i$ was feasible along time slots.
\vspace{-4mm}
\begin{figure}[thpb]
      \centering
      \includegraphics[height=1.62in,width=2.31in]{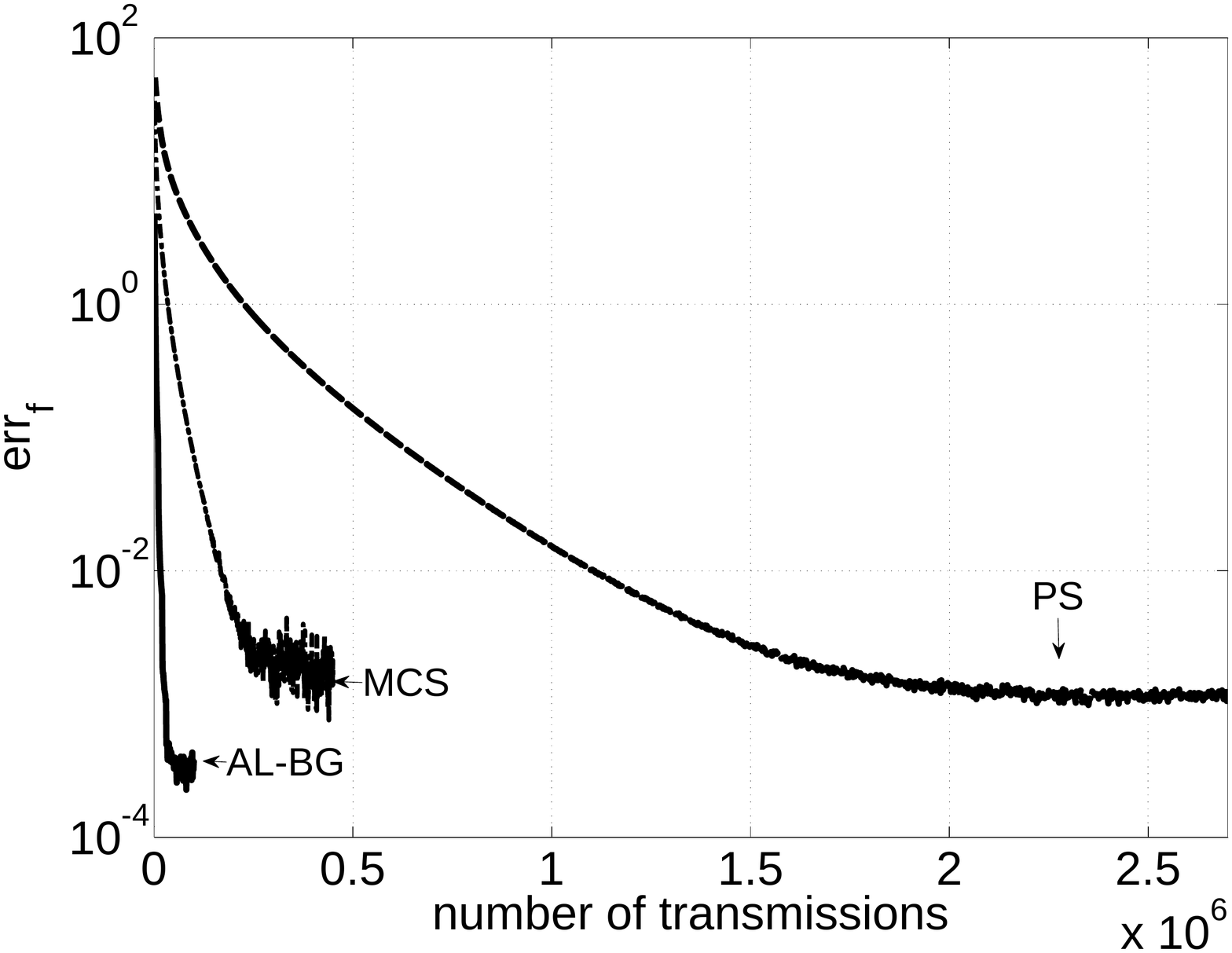}
      \includegraphics[height=1.62in,width=2.31in]{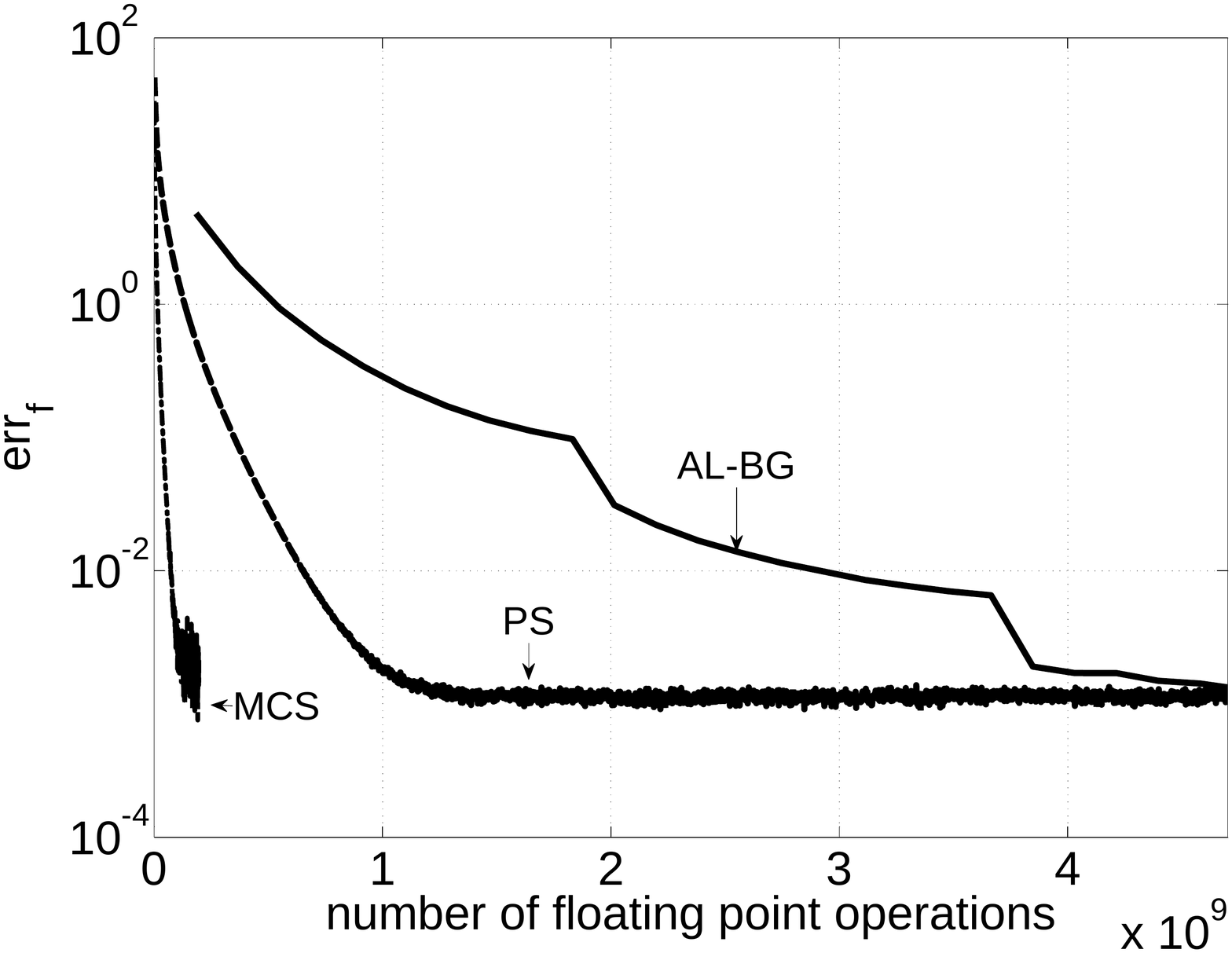}
      \includegraphics[height=1.62in,width=2.31in]{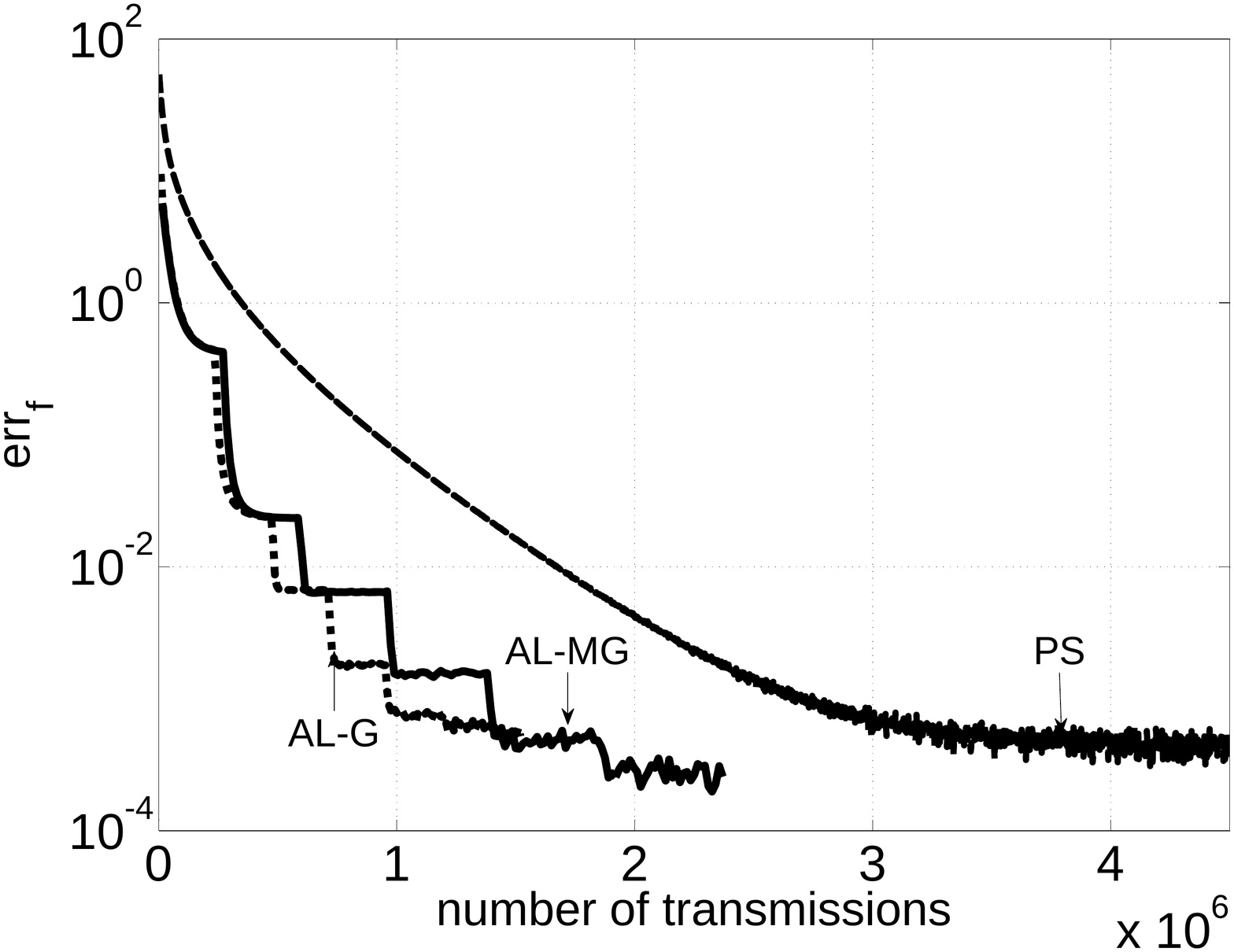}
      \includegraphics[height=1.62in,width=2.31in]{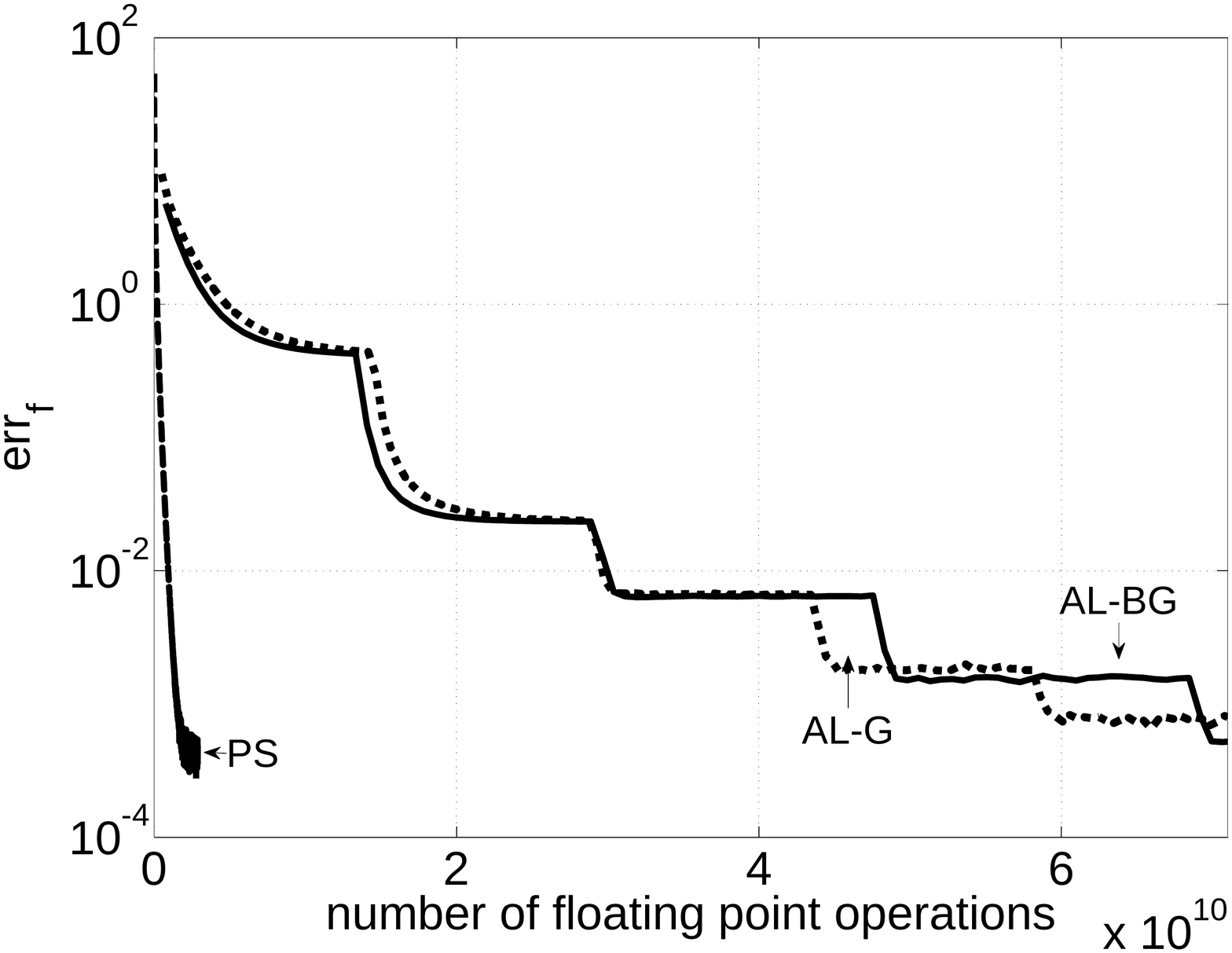}
      \label{slika-regression}
      \caption{Performance of \textbf{AL--BG}, MCS, and PS on a static network (top figures,) and the \textbf{AL--G}, \textbf{AL--MG} and PS algorithms on a random network
       (bottom figures.) Left: total number of transmissions; Right: total number of floating point operations.}
\end{figure}
\vspace{-4mm}

\vspace{-7mm}
\subsection{Cooperative spectrum sensing for cognitive radio networks}
\label{Cognitive_Radio}
We now consider cooperative spectrum sensing for cognitive radio networks.
Cognitive radios are an emerging technology for improving the efficiency of usage of the radio spectrum.
(For a tutorial on cognitive radios see, e.g.,~\cite{CR_tutorial}.) We focus here on the cooperative spectrum sensing approach that has been studied in~\cite{bazerque_lasso, bazerque_sensing}. Suppose that $N_r$
 cognitive radios, located at $x_r$ positions in 2D space, cooperate to determine: 1) the spatial locations; and
 2) the power spectrum density (PSD) of primary users. Primary users can be located on $N_s$ potential locations, $x_s$, on $\sqrt{N_s} \times \sqrt{N_s}$ square grid (See Figure 3, top, in~\cite{bazerque_sensing}.)
 For brevity, we omit the details of the problem setup; we refer to reference~\cite{bazerque_lasso},
 subsection~II-A,
 for the problem setup, and section~II (eqn.~\eqref{eqn_primal_reformulated_with_y_s}) in the same reference, for the
 Lasso optimization problem of estimating the locations and the PSD of primary users.
 This (unconstrained) optimization problem in eqn.~\eqref{eqn_primal_reformulated_with_y_s} in~\cite{bazerque_lasso}
  fits the generic framework in~eqn.~\eqref{eqn_original_primal_problem}; thus, our algorithms \textbf{AL--G},
   \textbf{AL--MG} and \textbf{AL--BG} apply to solve the problem in eqn.~\eqref{eqn_primal_reformulated_with_y_s} in~\cite{bazerque_lasso}. Throughout,
  we use the same terminology and notation as in~\cite{bazerque_lasso}.  We now detail the simulation parameters. The number of potential sources is
 $N_s=25$; they are distributed on a regular $5 \times 5$ grid over the square surface of $4\mathrm{km}^2$.
 Channel gains $\gamma_{sr}$ are modeled as $\gamma_{sr} = \min \left\{ 1, \,\,\frac{A}{\|x_s-x_r\|^a} \right\}$, with $A=200$ [meters] and $a$=3. The number of basis rectangles is $N_b=6$,
 and the number of frequencies at which cognitive radios sample PSD is $N_f=6$.
 There are 3 active sources; each source transmits at 2 out of $N_b=6$ possible frequency bands.
  After some experimentation, we set the Lasso parameter $\lambda$ (see eqn.~\eqref{eqn_primal_reformulated_with_y_s}
  in~\cite{bazerque_lasso}) to $\lambda=1$; for a distributed algorithm to optimally~set~$\lambda$,~see~\cite{bazerque_lasso}. We consider the supergraph with $N_r=20$ nodes (cognitive radios) and
$|E|=46$ undirected edges ($92$ arcs.) Nodes are uniformly distributed on a unit 2km$\times$2km square
and the pairs of nodes with distance smaller than $r=$750m
are connected.

For static networks, we compare \textbf{AL--BG} (our algorithm) with MCS, PS, and an algorithm in~\cite{bazerque_lasso}. Reference~\cite{bazerque_lasso} proposes three (variants of AD-MoM type algorithms, mutually differing in: 1) the total number of primal and dual variables maintained by each node (cognitive radio); 2) the method by which nodes solve local optimizations for primal variable update (These problems are similar to~\eqref{eqn_minimize_with_resp_x_i-simplif}.) We compare \textbf{AL--BG} with the DCD-Lasso variant, because it has the same number of primal and dual variables as \textbf{AL--BG} and a smaller computational cost than the alternative DQP-Lasso variant. With \textbf{AL--BG}, we use an increasing sequence of AL penalty parameters, $\rho_{t}=K_{\rho}\,A_{\rho}^t+C_{\rho}$, $t=0,1,...$, with $K_{\rho}=1$, $A_{\rho}=1.15$ and $C_{\rho}=3$. With DCD-Lasso, we used fixed
$\rho=\rho_t$, as in~\cite{bazerque_lasso,bazerque_sensing}.\footnote{It may be possible
 to improve on the speed of DCD-Lasso by selecting appropriate time varying $\rho=\rho_t$; this is outside of our paper's scope.}
We solve the local problems in \textbf{AL--BG} (eqn.~\eqref{eqn_minimize_with_resp_x_i-simplif}), \textbf{AL--G} and \textbf{AL--MG} (eqn.~\eqref{eqn_solving_wrt_x_i-simplified},) by an efficient block coordinate method in~\cite{bazerque_lasso} (see eqn.~(13)~in~\cite{bazerque_lasso}.)
%
%
%
For the networks with link failures, we have compared our \textbf{AL--G} and \textbf{AL--MG}
 algorithms with PS (in~\cite{asu-new-jbg,nedic_T-AC,nedic_novo,asu-random}.) We briefly comment on the results. Both \textbf{AL--G}
   and \textbf{AL--MG} converge to a solution, in the presence of
   link failures as in~VI-A; they converge
    slower than the PS algorithm, both in terms of communication
     and computational cost.

\mypar{Results for static network}
Figure 3 (left) plots $\textbf{err}_f$ for PS, MCS, DCD-Lasso, and \textbf{AL--BG} versus the number of transmissions (at all nodes.) \textbf{AL--BG} shows improvement over the other algorithms. To achieve the precision of $\textbf{err}_f \leq 0.044$,
 \textbf{AL--BG} requires about $5 \cdot 10^4$ transmissions; MCS $20 \cdot 10^4$
  transmissions; DCD-Lasso $25 \cdot 10^4$ transmissions; PS $50 \cdot 10^4$ transmissions. Limiting error for PS is $0.027$ (not visible in the plot.)
   Note also that DCD-Lasso and PS saturate at a larger error than \textbf{AL--BG} and MCS.
  \vspace{-2mm}
   \begin{figure}[thpb]
      \centering
      \includegraphics[height=1.6in,width=2.3in]{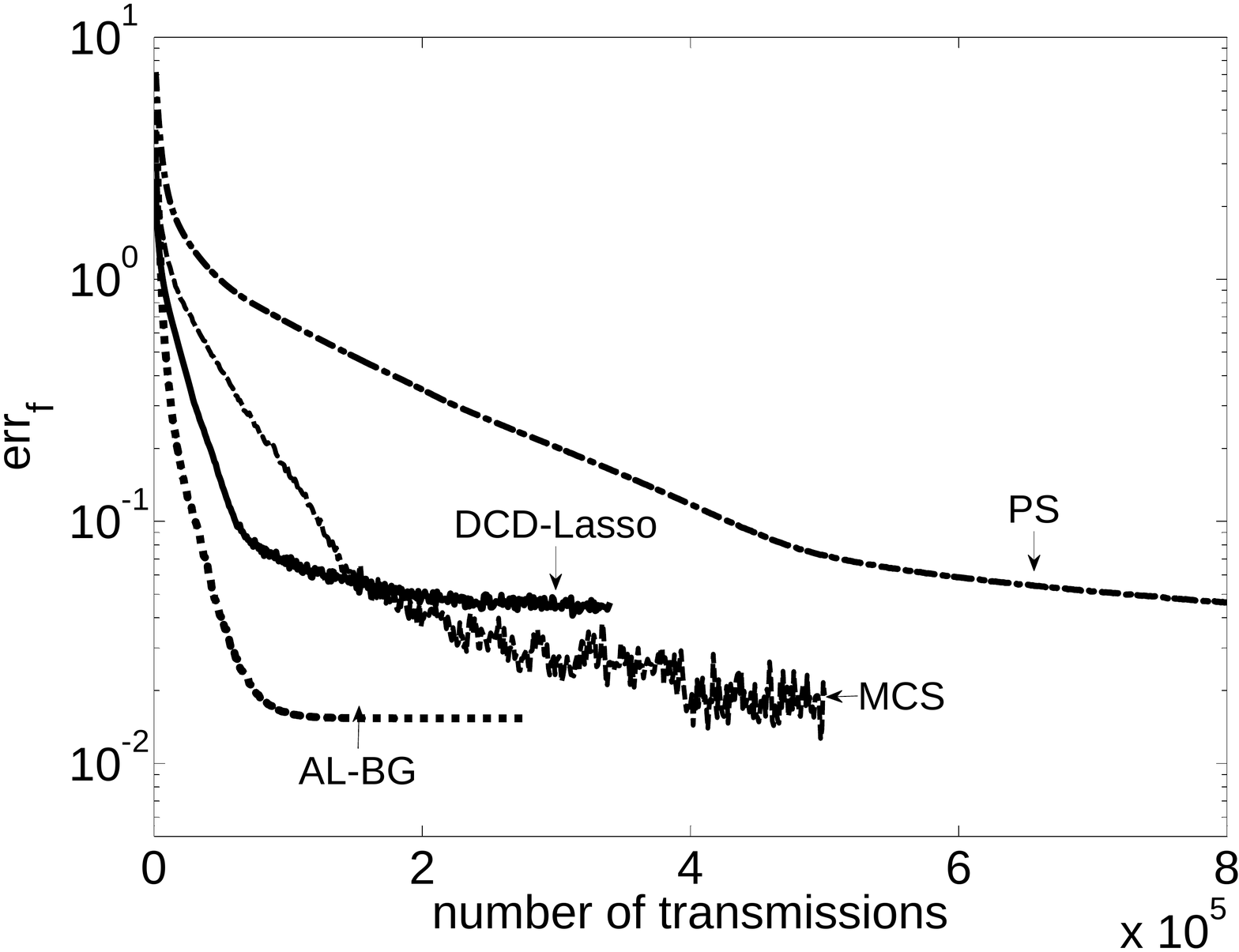}
      \includegraphics[height=1.6in,width=2.3in]{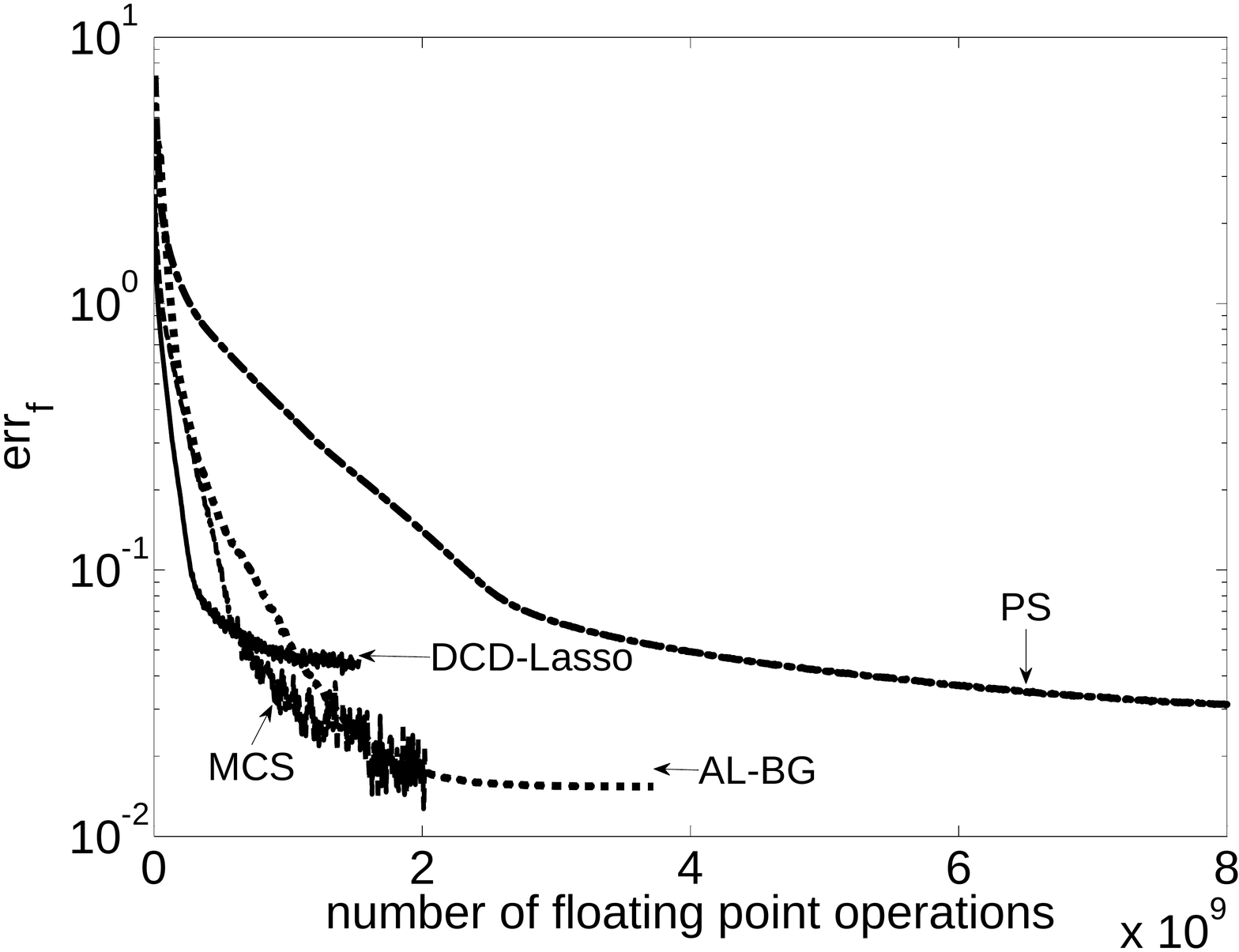}
      \caption{Performance of \textbf{AL--BG}, DCD-Lasso, PS and MCS algorithms on static CR network. Left: total number of transmissions (cumulatively, at all nodes). Right: total number of floating point operations (cumulatively, at all nodes.)}.
   \end{figure}
   \vspace{-0mm}
   Figure 3 (right) plots the $\textbf{err}_f$ for the PS, MCS, DCD-Lasso, and \textbf{AL--BG} algorithms
    versus the total number of floating point operations. \textbf{AL--BG}, MCS and DCD-Lasso show similar performance, while PS is slower.
\vspace{-3mm}
\section{Conclusion}
We studied cooperative optimization in networked systems, where each node obtains
 an optimal (scalar or vector) parameter of common interest, $x=x^\star$.
 Quantity $x^\star$ is a solution to the optimization problem where the objective
 is the sum of private convex objectives    at each node, and each node has a private
 convex constraint on $x$. Nodes utilize gossip to communicate through
a generic connected network with failing links. To solve this network problem, we proposed a novel distributed,
decentralized algorithm, the \textbf{AL--G} algorithm. \textbf{AL--G}
 handles a very general optimization problem with private costs, private constraints, random networks,
  asymmetric link failures, and  gossip communication.

This contrasts with existing augmented Lagrangian primal-dual methods
that handle only static networks and synchronous communication, while, as mentioned,
 the \textbf{AL–-G} algorithm handles
random networks and uses gossip communication.
In distinction with existing primal subgradient algorithms
 that essentially handle only symmetric link failures,
 \textbf{AL--G} handles asymmetric link failures.

\textbf{AL--G} updates the dual variables synchronously via a standard
 method of multipliers, and it updates the primal variables
 via a novel algorithm with gossip communication, the P--AL--G algorithm.
 P--AL--G is a nonlinear Gauss-Seidel type algorithm with
  random order of minimizations. Nonlinear Gauss-Seidel
   was previously shown to converge only
    under the \emph{cyclic} or the \emph{essentially cyclic} rules,~\cite{Tseng_new,BertsekasOptimization}; we prove convergence of    P--AL--G, which has a \emph{random} minimization order. Moreover, our
    proof is different from standard proofs for nonlinear Gauss-Seidel, as
    it uses as main argument the expected decrease in the objective function after
    one Gauss-Seidel step. We studied
    and proved convergence of  two variants of \textbf{AL--G}, namely, \textbf{AL--MG} and \textbf{AL--BG}. An interesting
     future research direction is to develop a fully asynchronous primal-dual algorithm that updates
     both the dual and primal variables asynchronously.

The \textbf{AL--G} algorithm is a generic tool to solve a wide range of problems
in networked systems; two simulation examples, $l_1$--regularized logistic regression for classification, and cooperative spectrum sensing for cognitive radios, demonstrated the applicability and effectiveness of \textbf{AL--G} in applications.
 \vspace{-3mm}
\appendix
\section{Appendix}
\label{AppendixSection}
\vspace{-3mm}
\textbf{Proof of Lemma~8.} We first need a standard result from topology (proof omitted for brevity.)
\begin{lemma}
Let $\mathcal{X}$ and $\mathcal{Y}$ be topological spaces, where $\mathcal{Y}$
 is compact. Suppose the function: $\kappa: \mathcal{X} \times \mathcal{Y} \rightarrow {\mathbb R}$
  is continuous (with respect to the product topology on $\mathcal{X} \times \mathcal{Y}$
   and the usual topology on $\mathbb R$; $\times$ denotes Cartesian product.) Then, the function $\gamma: \mathcal{X}
    \rightarrow {\mathbb R}$, $\gamma(a):=\inf\{\kappa(a,b):\,b \in \mathcal{Y}\}$~is~continuous.
\end{lemma}
\begin{proof}[Proof of Lemma~8]
Denote by $\mathcal{P}_i:\, {\mathbb R}^{m(N+2M)} \rightarrow {\mathbb R}^{m}$ the
projection map $\mathcal{P}_i(z)
= z_i$, $i=1,...,N+2M$. Further, denote by $\mathcal{P}_i(\Gamma(\epsilon)):= $$\left\{ z_i \in {\mathbb R}^{m}:\,\,z_i = \mathcal{P}_i(z),
 \,\,\mathrm{for\,\,some\,\,}z \in \Gamma(\epsilon) \right\}$.
 The set $\mathcal{P}_i(\Gamma(\epsilon))$ is compact, for all $i=1,...,N+2M$, because
  the set $\Gamma(\epsilon)$ is compact. Consider now the set ${\mathbb R}^{m(N+2M)} \supset C_{\epsilon}
 := \mathcal{P}_1(\Gamma(\epsilon)) \times \mathcal{P}_2(\Gamma(\epsilon)) \times\,...\,
 \times \mathcal{P}_{N+2M}(\Gamma(\epsilon))$,
 where the symbol $\times$ denotes the Cartesian product of the sets. Clearly, $C_{\epsilon} \supset {\Gamma}_{\epsilon}(B)$. We will show that $L^i$
is continuous on $C_{\epsilon}$, i.e., that $L^i: C_{\epsilon} \rightarrow {\mathbb R}$
 is continuous, which will imply the claim of Lemma~8. Recall the definition of $L^i$ in eqn.~\eqref{eqn_L^i(z)}.
 It is easy to see that the minimum in eqn.~\eqref{eqn_L^i(z)} is attained on the
  set $\mathcal{P}_i \left( \Gamma(\epsilon) \right)$, i.e., that $L^i(z) = \min_{w_i \in \mathcal{P}_i \left( \Gamma(\epsilon) \right)} L \left( z_1,z_2,...,z_{i-1},w_i,z_{i+1},...,z_{N+2M} \right).$
  Thus, by Lemma~12, and because the function $L: {\mathbb R}^{m(N+2M)}\rightarrow \mathbb R$ is
  continuous, the function $L^i: \,\, \mathcal{P}_1(\Gamma(\epsilon)) \times
  \,...\, \times \mathcal{P}_{i-1}(\Gamma(\epsilon)) \times \mathcal{P}_{i+1}(\Gamma(\epsilon))
  \times...\times \mathcal{P}_{N+2M}(\Gamma(\epsilon)) \rightarrow {\mathbb R}$ is continuous. But this means that
   $L^i: C_{\epsilon} \rightarrow {\mathbb R}$ is also continuous.
\end{proof}

\vspace{-1.mm}

\textbf{Convergence proof of the P--AL--MG algorithm.}
We first introduce an abstract model of the P--AL--MG algorithm.
First, we impose an additional assumptions that the link failures
 are spatially independent, i.e., the Bernoulli states $A_{ij}(k)$
 and $A_{lm}(k)$ of different links at time slot $k$ are independent. Define the sets
$Y(\Omega_i):=\left\{ y_{ji}:\,\,j \in \Omega_i \right\}$ and the class
$Y(O_i):=\left\{ y_{ji}:\,\,j \in O_i \right\}$, where $O_i \subset \Omega_i$.
  One distinct set $Y(O_i)$ is assigned to each distinct subset $O_i$ of $\Omega_i$. (Clearly, $Y(\Omega_i)$ belongs to a class of sets $Y(O_i)$, as $\Omega_i$
   is a subset of itself.) With P--AL--MG, at iteration $k$, minimization is performed either with respect to $x_i$, $i \in \{1,...,N\}$, or with respect to some $Y(O_i)$. If none of the neighbors of node $i$ receives successfully a message, then iteration $k$ is void. Define the following collection of the subsets of primal variables:
 $\Pi:=\{\{x_1\},...,\{x_N\},Y(\Omega_1),...,Y(\Omega_N)\}$. Collection $\Pi$ constitutes a partition of the set of
  primal variables; that is, different subsets in $\Pi$ are disjoint and their union contains all primal variables. Further,
   denote each of the subsets
   $\{x_i\}$, $Y(O_i)$, $Y(\Omega_i)$, with appropriately indexed $Z_s$, $s=1,...,S$. Then, with P--AL--MG, at time slot $k$, $L(z)$ is optimized with respect to one $Z_s$, $s=1,...,S$. Define $\xi(k)$,
    $k=1,2,...$, as follows: $\xi(k)=s$, if, at time slot $k$, $L(z)$ is optimized with
    respect to $Z_s$; $\xi(k)=0$, if, at $k$, no variable gets updated--when all transmissions at time slot $ k$ are unsuccessful. Denote by $P(Z_s) =
    \mathrm{Prob} \left(\xi(k)=s \right)$. Under
     spatial independence of link failures, $P(Z_s)$ can be
     shown to be strictly positive for all $s$. It can be shown that $\xi(k)$ are i.i.d. Consider now~\eqref{eqn_opt_problem_L(z)} and P--AL--MG. All results for P--AL--G remain valid for P--AL--MG also--only the expressions for the expected decrease of $L(\cdot)$ per iteration, $\psi(z)$, (Lemma~7), and the proof of Lemma~8 change. Denote by $L^{(Z_s)}(z)$ the optimal value after minimizing $L(\cdot)$ with respect to $Z_s$ at point $z$ (with the other blocks $z_j$, $z_j \notin Z_s$, fixed.)~Then:
$
\psi(z) = \sum_{s=1}^{S} P(Z_s) \left( L^{(Z_s)} - L(z)  \right).
$
Recall $\phi(z)=-\psi(z)$ and the set $\Gamma(\epsilon)$, for some $\epsilon>0$.
Lemma~8 remains valid for P--AL--MG. To see this, first remark that $\phi(z) \geq 0$, for all $z \in F$. We want to show that $\phi(z)>0$, for all $z \in \Gamma(\epsilon)$. Suppose not. Then, $L(z)=L^{(Z_s)}(z)$, for all $Z_s$, $s=1,...,S$. Then, in particular, $L(z)=L^{(Z_s)}(z)$, for all $Z_s$ in the partition $\Pi$. Because $P(Z_s)>0$, $\forall s$, this implies that the point $z$ is block-optimal (where now, in view of Definition~\ref{definition-block-optimal}, $Z_s$ is considered a single block). By
Assumption~3,
$z$ is also optimal, which contradicts $z \in \Gamma(\epsilon)$. Thus, $\phi(z)>0$ for all $z \in \Gamma(\epsilon)$. The proof now proceeds as with the proof of Lemma~8 for algorithm P--AL--G.
%
%
\vspace{-1.mm}

\textbf{Convergence proof of the P--AL--BG algorithm.}
P--AL--BG is completely equivalent to P--AL--G, from the optimization
 point of view. P--AL--BG can be modeled in the same way as in Alg.~2, with a difference that, with P--AL--BG: 1) there is a smaller number ($=N$) of primal variables: $z_i:=x_i$, $i=1,...,N$;
   and 2) $\mathrm{Prob}(\zeta(k)=0)=0$. Thus, the analysis in section~V is also valid for~P--AL--BG.
\vspace{-3mm}
\bibliographystyle{IEEEtran}
\bibliography{IEEEabrv,bibliography_distributed_optimization}

\end{document}